\documentclass{article}
\usepackage[utf8]{inputenc}
\usepackage[utf8]{inputenc}
\usepackage{amsmath,amssymb}
\usepackage{amsthm}
\usepackage{algorithm}
\usepackage{algorithmic}
\usepackage{graphicx}
\usepackage{url}
\usepackage{amsmath}
\usepackage[left=1.5cm, right=1.5cm, top=1.5cm, bottom=2cm]{geometry}
\newtheorem{definition}{Definition}
\newtheorem{assumption}{Assumption}
\newtheorem{remark}{Remark}
\newtheorem{lemma}{Lemma}
\newtheorem{theorem}{Theorem}
\newtheorem{proposition}{Proposition}

\usepackage{authblk}
\usepackage{hyperref}
\hypersetup{colorlinks=true,urlcolor=blue}

\title{Stochastic Average Consensus Filtering and Distributed State Estimation for Boolean Control Networks}
\author[1]{Rong Yang\thanks{Email: \url{42326006@smail.swufe.edu.cn}}}
\author[2]{Chi Huang$^\ast$\thanks{$^\ast$ Corresponding author, Email: \url{huangchi@swufe.edu.cn}}}

\affil[1]{School of Mathematics, Southwestern University of Finance and Economics}
\affil[2]{School of Computing and Artificial Intelligence, Southwestern University of Finance and Economics}

\date{}

\begin{document}
	\maketitle
	
	\begin{abstract}
		This paper addresses the distributed multi-sensor fusion state estimation and consensus filtering for  Boolean control networks (BCNs). Existing centralized multi-sensor estimation schemes for stochastic BCNs have limitations of high communication costs and single-point failures, and continuous-state consensus algorithms are difficult to extend to discrete logical systems. By integrating probability measure transformation, semi-tensor product and stochastic approximation, a distributed stochastic average consensus filter is proposed. Moreover, the almost sure convergence of the algorithm is proved by martingale convergence theorem and perturbed stochastic Lyapunov functions. The proposed framework realizes global state estimation via local communication, avoiding the defects of centralized architectures.
	\end{abstract} 
	
	
	\section{Introduction}
	
	Boolean Networks (BNs) were first proposed by Kauffman \cite{Kauffman1969} in 1969 as a classic logical dynamic model for studying Gene Regulatory Networks (GRNs) and cellular metabolic processes, garnering widespread attention across multiple disciplines. In a BN, each node evolves between two discrete states activation (1) and inhibition (0). Subsequently, with the integration of control theory, researchers introduced external control inputs to BNs, leading to the development of BCNs \cite{BCN_intro}, which provide a powerful mathematical framework for investigating targeted interventions, drug treatment strategies, and therapeutic pathways for diseases in complex biological systems.
	
	Traditional mathematical analysis of logical dynamic systems has encountered significant challenges owing to their strong nonlinearity. To address this bottleneck, Cheng et al. proposed the Semi-tensor Product (STP) \cite{Cheng2009STP}, which successfully transforms discrete logical evolution equations into an equivalent Algebraic State Space Representation. This breakthrough enabled the application of matrix analysis from modern control theory to solve core problems of BNs/BCNs, such as stability \cite{Stability_BN}, controllability \cite{Controllability_BCN}, optimal control \cite{OptimalControl_BCN}, and fault detection \cite{FaultDetection_BCN}, laying the foundation for subsequent stochastic extension research.
	
	In real world biological and engineering systems, network dynamics are rarely purely deterministic. Molecular thermal motion, inherent randomness in gene expression, and external environmental fluctuations inevitably introduce process and measurement noise \cite{Noise1, Noise2}. To accurately capture such uncertainty, advanced stochastic extensions of BNs have been developed, including Probabilistic BNs \cite{PBN_intro}, Markov Jump BNs \cite{MJBN_intro}, and Stochastic BNs \cite{SBN_intro}. State estimation for these stochastic systems is pivotal for predicting biological behaviors and implementing effective control strategies. For instance, estimating the state of the p53-MDM2 feedback network is critical for cancer treatment \cite{p53_MDM2}.
	
	To mitigate the impact of noise and improve estimation reliability, multi-sensor fusion has emerged as a key research direction \cite{MultiSensor1, MultiSensor2}. By integrating redundant information from multiple sensors, multi-sensor fusion can enhance estimation accuracy and robustness compared to single-sensor systems, making it indispensable for high precision applications such as power grids, aircraft, and body sensor networks \cite{MultiSensorApp1, MultiSensorApp2, MultiSensorApp3}. For SBNs, Li and Tang \cite{Li2022MultiSensor} pioneered multi-sensor fusion Boolean Bayesian filtering, proposing a recursive algorithm to compute prior and posterior state beliefs without dimension expansion. Their work demonstrated that multi-sensor fusion significantly improves state estimation confidence compared to single sensor schemes \cite{Li2022MultiSensor}. Other existing studies have explored optimal estimation via Bayesian filtering \cite{Bayesian_SBN}, maximum-likelihood methods \cite{ML_SBN}, and hidden Markov modeling \cite{HMM_BCN}, but most rely on centralized architectures.
	
	A critical limitation of current multi-sensor fusion and state estimation research for SBNs is the centralized processing paradigm \cite{Li2022MultiSensor, CentralizedLimitation}. In centralized frameworks, a central processor collects measurement data from all sensors for global computation, which faces inherent drawbacks, such as high single-point failure risk, severe communication bandwidth bottlenecks, and exponential computational complexity with increasing network scale and sensor quantity \cite{DistributedConsensus}. This is particularly problematic for large-scale GRNs or distributed industrial systems, where sensor nodes are geographically dispersed, communication resources are limited, and the number of sensors is substantial. In contrast, distributed state estimation enables each sensor node to complete global estimation tasks using only local observation data and limited communication with neighboring nodes, offering superior robustness, scalability, and energy efficiency \cite{DistributedConsensus}.
	
	Existing multi-sensor fusion methods \cite{Li2022MultiSensor} still depend on data aggregation, failing to leverage the advantages of distributed architectures. Specifically, achieving consistent optimal estimation of BCN states across the entire sensor network using limited local communication, while integrating multi-sensor information to mitigate noise impacts, has not been addressed in existing literature. Current multi-sensor state estimation \cite{Li2022MultiSensor,26} adopts centralized processing, suffering from poor scalability, high communication costs, and single-point failure risks in large-scale sensor networks. Existing distributed methods lack effective integration of multi-sensor redundant information to minimize estimation errors, as achieved by centralized multi-sensor fusion.

	Motivated by these gaps, this paper focuses on distributed multi-sensor fusion state estimation and consensus filtering for disturbed BCNs in sensor networks. We deeply integrate hidden Markov modeling, probability measure transformation, multi-sensor data fusion principles \cite{Li2022MultiSensor,26}, and distributed stochastic approximation algorithms to propose a comprehensive theoretical and algorithmic framework. This paper's contributions are summarized as follows
	\begin{enumerate}
		\item To overcome computational challenges from strongly nonlinear observations and heterogeneous sensor data, we introduce the Radon-Nikodym derivative for change of Measure. Under the new measure, we derive recursive equations for each node's conditional state expectation based on local sensor observations.
		\item Addressing communication constraints, we propose a distributed algorithm based on stochastic approximation. Each sensor node updates its global state estimate by exchanging local estimate results with topological neighbors, avoiding centralized bottlenecks.
		\item Using the martingale convergence theorem and perturbed stochastic Lyapunov functions, we prove almost sure practical convergence of the proposed algorithm. Under graph connectivity and step-size conditions, each sensor node's estimation error is uniformly ultimately bounded by an explicitly quantified constant, with the bound governed by the network algebraic connectivity and sensor heterogeneity. The radius of this bound can be systematically reduced by enhancing the network algebraic connectivity.
	\end{enumerate}
	
	Notations: Set $\mathcal{D}$ is denoted as $\mathcal{D} = \{0,1\}$. A $p \times q$-dimensional logical matrix $L$ can be represented by $L=[\delta_{p}^{k_{1}}\delta_{p}^{k_{2}}\ldots\delta_{p}^{k_{q}}]$, where $k_{i} \in [1:p]$, for $i \in [1:q]$. 
	It is also represented by $L=\delta_{p}[k_{1}k_{2}\ldots k_{q}]$ for simplicity. 
	All $p \times q$-dimensional logical matrices construct set $\mathcal{L}_{p\times q}$. For matrix $A$, $\operatorname{Col}_{j}(A)$ be the $j$th column, $\operatorname{Row}_{i}(A)$ be the $i$th row, and $A_{p,q}$ corresponds to the element of the $p$th row and the $q$th column. And $|S|$ represents the number of elements in set $S$. $\mathbf{1}_{n}(\mathbf{0}_{n})$ is an $n$-dimensional vector with all entries being $1 (0)$. For an $n$-dimensional vector $x$, $x_{i}$ represents the $i$th entry of vector $x$. Let $\lfloor x \rfloor$ denote the floor of x, i.e., the greatest integer less than or equal to x. $\langle \cdot , \cdot \rangle$ represents the inner product of vectors.
	
	\section{PRELIMINARIES}
	
	\begin{definition}
		The semi-tensor product of two matrices $A \in \mathbb{R}^{m \times n}$ and $B \in \mathbb{R}^{p \times q}$ is defined as follows:
		\begin{equation}
			A \ltimes B = \left(A \otimes I_{\alpha/n}\right)\left(B \otimes I_{\alpha/p}\right)
		\end{equation}
		where $\alpha = \mathrm{lcm}(n,p)$ denotes the least common multiple of $n$ and $p$, and $\otimes$ represents the tensor  product.
	\end{definition}
	
	Given logical arguments $x_1, \ldots, x_n \in \Delta_2$ and a logical function $f(x_1, \ldots, x_n)$, there exists a unique structure matrix $L_f \in \mathcal{L}_{2 \times 2^n}$ such that the logical function admits the following multilinear representation
	$f(x_1, \ldots, x_n) = L_f x_1 \ltimes \cdots \ltimes x_n$.
	Let $x = \ltimes_{i=1}^{n} x_i$, then the dynamic equation of node $i$ can be formulated using its structure matrix $L_{f_i}$ as $x_i(t+1) = L_{f_i} x(t)$. To characterize the relationship between $x(t+1)$ and $x(t)$, we first introduce the following operation.
	\begin{definition}
		For matrices $A \in \mathbb{R}^{m \times r}$ and $B \in \mathbb{R}^{n \times r}$, the Khatri--Rao product of $A$ and $B$ is defined by
		\begin{equation}
			A * B = \left(\mathrm{Col}_1(A) \ltimes \mathrm{Col}_1(B), \ldots, \mathrm{Col}_r(A) \ltimes \mathrm{Col}_r(B)\right).
		\end{equation}
	\end{definition}
	Accordingly, $x(t+1) = Lx(t)$, where $L = L_{f_1} * \cdots * L_{f_n}$. 
	Let the swap matrix be defined as $W_{[m,n]} = \left[ I_n \otimes \delta_m^1, \ldots, I_n \otimes \delta_m^m \right]$ and $D^{[m,n]}_r = I_m \otimes \mathbf{1}^\top_n$, then one has $x_i = D^{[2,2^{n-1}]}_r W_{[2,2^{i-1}]} x$.
	
	The communication topology of the sensor network is modeled by an undirected simple graph $\mathcal{G} = (\mathcal{N}, \mathcal{E})$, where
	\begin{itemize}
		\item $\mathcal{N} = \{1, 2, \dots, r\}$ with $r \in \mathbb{N}$ denotes the set of vertices, where vertex corresponds to a sensor node,
		\item $\mathcal{E} \subseteq \mathcal{N} \times \mathcal{N}$ represents the set of edges, where an unordered pair $(i, j) \in \mathcal{E}$ indicates a communication link between node $i$ and node $j$.
	\end{itemize}
	
	Throughout this paper, all graphs are assumed to be undirected and simple, i.e., no self-loops. Formally, for every edge $(i, j) \in \mathcal{E}$, we require $i \neq j$. The set of neighbors of node $j$ is defined as $\mathcal{N}_j \triangleq \{i \in \mathcal{N} \mid (i, j) \in \mathcal{E}\}$. 
	
	\section{PROBLEM FORMULATION}
	
	Consider the following disturbed BCN. For $i = 1, 2, \ldots, n$
	\begin{equation}
		X_i(k+1) = f_i(X(k), U(k), W(k))
	\end{equation}
	where $X_i(k) \in \mathcal{D}$ is the state value of node $i$, and the network state $X(k)$ is defined as $X(k) = (X_1(k), \ldots, X_n(k))^{\top} \in \mathcal{D}^n$.
	The control input is $U(k) = (U_1(k), \ldots, U_u(k))^{\top} \in \mathcal{D}^u$, and the system disturbance is $W(k) = (W_1(k), \ldots, W_{l_1}(k))^{\top} \in \mathcal{D}^{l_1}$,
	with $U_j(k) \in \mathcal{D}$, $j \in [1:u]$ and $W_t(k) \in \mathcal{D}$, $t \in [1:l_1]$. Here, $u$ and $l_1$ are the numbers of control inputs and disturbance inputs, respectively.
	The disturbance noise is assumed to be independent and identically distributed. Moreover, the state transition function $f_i(\cdot, \cdot, \cdot)$ maps $\mathcal{D}^{n+u+l_1}$ to $\mathcal{D}$.
	The output measurement dynamics for every sensor s $\in \mathcal{N}$ are given as follows
	\begin{equation}
		Y^s_j(k) = h^s_j(X(k), V^s(k)),  j\in[1:m_s] 
	\end{equation}
	where the output vector $Y^s(k) = (Y^s_1(k), \ldots, Y^s_{m_s}(k))^{\top} \in \mathcal{D}^{m_s}$, and $Y^s_j(k) \in \mathcal{D}$ denotes the observer measurement.
	$V^s(k) = (V^s_1(k), \ldots, V^s_{l^s_2}(k))^{\top} \in \mathcal{D}^{l^s_2}$ represents the output noise of system (3).
	Analogous to $f_i(\cdot, \cdot, \cdot)$, the output function $h^s_j(\cdot, \cdot): \mathcal{D}^{n+l^s_2} \to \mathcal{D}^{m_s}$.
	The initial state $X(0)$ and noises $W(k), V^s(k)$ are assumed to be mutually independent.
	Define a mapping from the set $\mathcal{D}$ to $\Delta_2$ such that $0 \mapsto \delta_2^2$ and $1 \mapsto \delta_2^1$.
	For node $i$, $f_i(X(k), U(k), W(k))$ corresponds to $\mathbb{E}\{f_i(x(k), u(k))\}$, where
	$x_i(k) = \delta_2^{2-X_i(k)}$, $u_j(k) = \delta_2^{2-U_j(k)}$, $x(k) = \ltimes_{i=1}^{n} x_i(k)$, $u(k) = \ltimes_{j=1}^{u} u_j(k)$.
	Consequently, there exists a unique probability matrix $F_i$ which is the expected structure matrix of node $i$. Then, by Markov property, system (3) can be transformed into
	\begin{equation}
		\mathbb{E}[x_i(k+1) | \mathcal{F}^U_k \cup \mathcal{F}^X_k] = \mathbb{E}[x_i(k+1) | u(k), x(k)] = F_i u(k) x(k),
	\end{equation}
	where $\mathcal{F}^U_k \triangleq \sigma(u(1),\ldots,u(k))$ and $\mathcal{F}^X_k \triangleq \sigma(x(0),\ldots,x(k))$ denote the $\sigma$-algebras generated by the control inputs and the states up to time $k$, respectively.
	Taking the STP of $\mathbb{E}[x_i(k+1) | u(k), x(k)]$ for $i=1$ to $n$, we obtain
	\begin{equation}
		\mathbb{E}[x(k+1) | \mathcal{F}^U_k \cup \mathcal{F}^X_k] =\mathbb{E}[x(k+1) | u(k), x(k)] = F u(k) x(k),
	\end{equation}
	where $F = F_1 * F_2 * \cdots * F_n$. This constructs a hidden Markov state control model.
	The probability matrix $F$ is partitioned into $2^u$ equal blocks as
	\begin{equation}
		F = \left[\mathrm{Blk}_1(F) \quad \mathrm{Blk}_2(F) \quad \cdots \quad \mathrm{Blk}_{2^u}(F)\right].
	\end{equation}
	For any $\alpha \in [1:2^u]$ and $i, j \in [1:2^n]$, under the control input $u(k) = \delta_{2^u}^{\alpha}$, we have
	\begin{equation}
		\left[\mathrm{Blk}_{\alpha}(F)\right]_{j,i} = \ P\left\{x(k+1) = \delta_{2^n}^j | x(k) = \delta_{2^n}^i, u(k) = \delta_{2^u}^{\alpha}\right\},
	\end{equation}
	where $\ P\{\cdot\}$ denotes the probability measure.
	We introduce the following state feedback control (SFC)
	\begin{equation}
		u(k) = K x(k),
	\end{equation}
	where K is a proper SFC gain.
	Let $y^s_j(k) = \delta_2^{2-Y^s_j(k)}$ and $y^s(k) = \ltimes_{j=1}^{m_s} y^s_j(k)$. Since $Y^s(k)$ is the measurement of the state $X(k)$, denote
	\begin{equation}
		\ P\left\{y^s(k) = \delta_{M_s}^j | x(k) = \delta_{N}^i\right\} = \left[H^s\right]_{j,i},
	\end{equation}
	where $M_s \triangleq 2^{m_s}$ and $N \triangleq 2^{n}$.
	With the output noise, the output feedback is characterized by the following hidden Markov model
	\begin{equation}
		\mathbb{E}\left\{y^s(k) | \mathcal{F}^U_k \cup \mathcal{F}^X_k \cup \mathcal{F}^{Y(s)}_k\right\} = \mathbb{E}\left\{y^s(k) | x(k)\right\} = H^s x(k),
	\end{equation}
	where $\mathcal{F}^{Y(s)}_k \triangleq \sigma({y^s(1),...,y^s(k)})$ denotes the $\sigma$-algebra generated by the observations up to time k. 
	
	With the system and observation models given by (6) and (11), the distributed state estimation problem addressed in this paper is formally stated as follows. Consider a sensor network consisting of \(r\) nodes modeled by an undirected graph \(\mathcal{G}=(\mathcal{N},\mathcal{E})\). Each node \(s\in\mathcal{N}\) serves a dual role. It acts as an observer that acquires the local observation sequence \(\{y^s(k)\}\) generated by the observation model, and simultaneously functions as a state estimator that computes a local inference of the global state. No central fusion center exists, and raw observations are never exchanged between nodes. At each time step \(k\), every node \(s\) is permitted to communicate exactly once with its immediate neighbors \(j\in\mathcal{N}_s\), exchanging only their current local estimate vectors. Under these communication constraints, the objective is to design a fully distributed  algorithm such that each node asymptotically tracks the network-wide average of all local conditional expectations with a bounded error. Specifically, the goal is to ensure
	\[	\limsup_{k \to \infty} \|\varepsilon_k\| \le c, \quad \mathbb{P}\text{-a.s.}\]
	where the error $\varepsilon_k$ represents the error between the consensus filter state vector and the average vector, and $c > 0$ is a finite constant determined by the network algebraic connectivity and the heterogeneity among sensor observation models. Based on this converged average estimate, each node further extracts an optimal estimate \(X^*_i(k)\). The main challenges lie in (i) recursively computing the local conditional expectation \(\mathbb{E}[x(k)\mid\mathcal{F}_k^{Y(s)}]\), and (ii) guaranteeing almost sure practical convergence of the distributed consensus process under stochastic system disturbances, observation noises, and the one-communication-per-step constraint. The subsequent section presents a solution that integrates semi-tensor product based recursive filtering with a stochastic approximation consensus scheme to overcome these difficulties.
	\section{MAIN RESULTS}
	\subsection{Stochastic Average Consensus Filtering}
	In this part, we utilize an approach to estimate the expectation of states. We introduce the following two assumptions.
	
	\begin{assumption}
		Let $Y_k \triangleq \{y^s(k)\}_{s \in \mathcal{N}}$ denote the collection of all sensor observations at time $k$. The observations of different sensors are mutually conditionally independent given the state $x(k)$. Specifically, for any finite subset of sensors $\{s_1, s_2, \ldots, s_m\} \subseteq \mathcal{N}$,
		\[
		\mathbb{P}\bigl(y^{s_1}(k), y^{s_2}(k), \ldots, y^{s_m}(k) \mid x(k)\bigr) = \prod_{i=1}^{m} \mathbb{P}\bigl(y^{s_i}(k) \mid x(k)\bigr).
		\]
	\end{assumption}
	
	\begin{assumption}
		The state process $\{x(k)\}$ is assumed to be geometrically ergodic under SFC (9).
	\end{assumption}
	
	\begin{remark}
		Under Assumption 2, the state process $\{x(k)\}$ admits a unique stationary distribution. 
		Accordingly, the extended Markov chain $\{(x(k), y^s(k))\}$ possesses a unique stationary distribution $\pi^s(\cdot,\cdot)$ defined on $\Delta_N \times \Delta_{M_s}$. 
		Denote by $y_i^s(\infty)$ the marginal density of the invariant measure $\pi^s$, where $y^s(\infty) = [y_i^s(\infty)]_{i \in [1:M_s]}$ and
		$ y_i^s(\infty) = \sum_{j=1}^{N} \pi^s(\delta_N^j, \delta_{M_s}^i).$
	\end{remark}
	
	In the following, we introduce an approach to estimate the conditional expectation of states. It is assumed that the dynamic of process $x(k)$ makes no difference under the new probability measure $P^s$ for each sensor $s$, that is
	\begin{equation}
		\mathbb{E}^s\bigl[x(k) \mid \mathcal{F}_{k-1}^ X \cup \mathcal{F}_{k-1}^ U \cup \mathcal{F}_{k-1}^{Y(s)}\bigr] = F u(k-1) x(k-1).
	\end{equation}
	However, the dynamic of process $y^s(k)$ is supposed to be independent of $x(k)$ and $u(k)$. Furthermore, it follows 
	\begin{equation}
		\mathbb{E}^s\bigl[\langle y^s(k), \delta_{M_s}^i \rangle \mid \mathcal{F}_k^X \cup \mathcal{F}_{k-1}^U \cup \mathcal{F}_{k-1}^{Y(s)}\bigr] = \frac{1}{M_s}
	\end{equation}
	for $i \in [1, M_s]$.
	Then we can denote the following parameters. Let
	\begin{equation}
		\Lambda_k^s = \prod_{t=1}^{k} \lambda_t^s, 
	\end{equation}
	where  $\lambda_k^s \triangleq M_s \sum_{j=1}^{M_s} \langle H^s x(k), \delta_{M_s}^j \rangle \langle y^s(k), \delta_{M_s}^j \rangle$. Then we have the following lemma.
	
	\begin{lemma}
		Under new probability measure $P^s$, the following two equations hold
		\begin{equation}
			\ \mathbb{E}^s\bigl[\lambda_k^s | \mathcal{F}_k^X \cup \mathcal{F}_{k-1}^U \cup \mathcal{F}_{k-1}^{Y(s)}\bigr] = 1 ,
		\end{equation} 
		\begin{equation}
			\ \mathbb{E}^s\bigl[\lambda_k^s | \mathcal{F}_{k-1}^X \cup \mathcal{F}_{k-1}^U \cup \mathcal{F}_{k-1}^{Y(s)}\bigr] = 1 .
		\end{equation}	
	\end{lemma}
	\begin{proof}
		See Appendix A.
	\end{proof}
	
	\begin{lemma}[{\cite[Th.~3.2 of Ch.~2]{41}}]
		The probabilistic space is denoted by $(\Omega, \mathcal{F}, P)$ and $\mathcal{G} \subset \mathcal{F}$ is a sub-$\sigma$-field. $\hat{P}$ is another probability measurement which is absolutely continuous with respect to $P$ and satisfies Radon--Nikodym derivative $dP/d\hat{P} = \Lambda$. For any $P$ integrable random variable $\phi$, one can obtain
		\begin{equation}
			\mathbb{E}[\phi | \mathcal{G}] = 
			\begin{cases}
				\dfrac{\hat{\mathbb{E}}[\Lambda\phi | \mathcal{G}]}{\hat{\mathbb{E}}[\Lambda | \mathcal{G}]}, & \hat{\mathbb{E}}[\Lambda | \mathcal{G}] > 0 \\[12pt]
				0, & \text{otherwise}.
			\end{cases}
		\end{equation}
	\end{lemma}
	
	\begin{theorem}
		The relationship (6) and (11), respectively, under the original probability $P$ can be recovered by Radon--Nikodym derivative, that is
		\begin{equation}
			\left.\frac{dP}{dP^s}\right|_{\mathcal{F}_k^X \cup \mathcal{F}_k^U \cup \mathcal{F}_k^{Y(s)}} = \Lambda_k^s.
		\end{equation}
	\end{theorem}
	\begin{proof}
		See Appendix B.
	\end{proof}
	
	Based on the above technical theorem, the relationship between the original probability $P$ and new probability $P^s$ is established. Next, we will solve the initial state estimation problem.
	
	\begin{theorem}
		The conditional expectation $\mathbb{E}[x(k) | \mathcal{F}_k^{Y(s)}]$ has the following recursive dynamics
		\begin{equation}
			\mathbb{E}\bigl[x(k+1) | \mathcal{F}_{k+1}^{Y(s)}\bigr] = \frac{\bigl[(H^s)^{\!\top} \mathcal{T}^s y^s(k+1)\bigr] \circ \bigl[FK\Phi_n \mathbb{E}[x(k) | \mathcal{F}_k^{Y(s)}]\bigr]}{\bigl\langle (H^s)^{\!\top} \mathcal{T}^s y^s(k+1), FK\Phi_n \mathbb{E}[x(k) | \mathcal{F}_k^{Y(s)}] \bigr\rangle}, 
		\end{equation}
		where ``$\circ$'' represents the Hadamard product, $\Phi_n \triangleq \delta_{2^{2n}}[1,2^n+2,2\cdot 2^n+3,...,(2^n-2)2^n+2^n-1,2^{2n}]$ and
		\[
		\mathcal{T}^s \triangleq \begin{bmatrix} M_s & & \\ & \ddots & \\ & & M_s \end{bmatrix}_{M_s \times M_s}
		\]  .
	\end{theorem}
	
	\begin{proof}
		See Appendix C.
	\end{proof}
	
	Having derived the recursive formulation for each node's local conditional state expectation, we next address the distributed fusion problem. While each sensor can independently compute its local posterior estimate from its own observation sequence, these heterogeneous local estimates cannot fully leverage the redundant information across the network. To enable collaborative state estimation with only local communication, we propose a distributed stochastic approximation algorithm that aims to track the network-wide average of all local conditional expectations. Formally, let the \(l\)-th component of \(\mathbb{E}[x(k) \mid \mathcal{F}_k^{Y(s)}]\) be denoted by \(e_k^s(l)\), i.e.,
	$e_k^s(l) \triangleq \bigl\langle \mathbb{E}[x(k) \mid \mathcal{F}_k^{Y(s)}], \delta_N^l \bigr\rangle, l \in [1:N].$
	Let \(\hat{\theta}_k^s(l)\) denote the consensus filter state at sensor \(s \in \mathcal{N}\) at time \(k\), which serves as the local estimate of the global average target
	$
	\theta_k^*(l) \triangleq \frac{1}{|\mathcal{N}|} \sum_{s \in \mathcal{N}} e_k^s(l).
	$
	Each sensor $s$ employs a stochastic approximation algorithm to estimate $\theta_k^*(l)$ using the input messages $e_k^{s'}(l), s' \in \mathcal{N}_s \cup \{s\}$, and consensus filter states $\hat{\theta}_k^s(l)$ only from its immediate neighbors. The state of each sensor $s \in \mathcal{N}$ is updated using the following algorithm
	\begin{equation}
		\begin{aligned}
			\hat{\theta}_k^s(l) &= \bigl(1 - \rho(1 + 2\operatorname{Row}_s(A) \cdot \mathbf{1})\bigr) \hat{\theta}_{k-1}^s(l)  + \rho\Bigl(\sum_{j=1}^{|\mathcal{N}|} a_{sj} \hat{\theta}_{k-1}^j(l) + \sum_{i=1}^{|\mathcal{N}|} a_{si} e_k^i(l) + e_k^s(l)\Bigr),
		\end{aligned} \label{eq:17}
	\end{equation}
	where $\rho$ is a fixed small scalar gain called step size, and $A = [a_{ij}]_{i,j \in \mathcal{N}}$ specifies the interconnection topology of the network which is assumed that $a_{ij} > 0$ for $j \in \mathcal{N}_i$ and is zero otherwise. Since the same algorithm is performed for every state $l$ in $\Delta_N$, to simplify the notation, henceforth we omit the notation $l$.
	
	\subsection{Convergence of the Algorithm and Optimal State Estimate}
	
	In this section, we study the convergence of the average consensus algorithm (20). The vectors are defined by $\hat{\theta}_k = [\hat{\theta}_k^s]_{s \in \mathcal{N}}$ and $e_k = [e_k^s]_{s \in \mathcal{N}}$, respectively. Define the error process $\{\varepsilon_k, k \in \mathbb{Z}^+\}$, where the error $\varepsilon_k$ is defined as
	\begin{equation}
		\varepsilon_k = \hat{\theta}_k - \theta_k^* \cdot \mathbf{1}
	\end{equation}
	which represents the error between the consensus filter state vector $\hat{\theta}_k$ and the average $\theta_k^*$. Let $\xi_k \triangleq (e_k, e_{k-1})$ adapted to $\sigma(Y_k , Y_{k-1})$ denotes the extended data at time k.

	\begin{lemma}
		For a given sequence $\{e_k, k \in \mathbb{Z}^+\}$, the error vector $\varepsilon_k$ evolves according to the following stochastic approximation algorithm
		\begin{equation}
			\varepsilon_{k+1} = \varepsilon_k + \rho \bar{\varepsilon}(\varepsilon_k, \xi_{k+1}), \quad k \in \mathbb{Z}^+ 
		\end{equation}
		where $\bar{\varepsilon}(\cdot, \cdot)$ is defined by
		\begin{equation}
			\bar{\varepsilon}(\varepsilon_k, \xi_{k+1}) \triangleq \Lambda \varepsilon_k + \varGamma\bigl(e_{k+1} - \tfrac{1}{|\mathcal{N}|}\mathbf{1}\mathbf{1}^\top e_k\bigr) - \tfrac{1}{|\mathcal{N}|\rho}\mathbf{1}\mathbf{1}^\top(e_{k+1} - e_k). 
		\end{equation}
		The matrices $\Lambda, \varGamma$ are defined by $\Lambda \triangleq \operatorname{diag}\bigl[-(1 + 2\operatorname{Row}_s(A) \cdot \mathbf{1})\bigr]_{s \in \mathcal{N}} + A$ and $\varGamma = I + A$, respectively.
	\end{lemma}
	In the following, we define a continuous-time interpolation $\mathcal{I}(t)$ of the error sequence $\{\varepsilon_k\}$ in terms of the step size $\rho$, that is,
	\[
	\mathcal{I}(t) \triangleq 
	\begin{cases}
		\varepsilon_k, & 0 \leq k\rho \leq t < (k+1)\rho, \\
		\varepsilon_0, & t < 0,
	\end{cases}\quad
	and  \ \mathcal{I}_n(t) \triangleq \mathcal{I}(t+n\rho).
	\]
	Define the mean vector field $\bar{\varepsilon}(\cdot)$ by
	\begin{equation}
		\bar{\varepsilon}(\eta) \triangleq \lim_{k \to \infty} \mathbb{E}_\eta\bigl[\bar{\varepsilon}(\eta, \xi_k)\bigr], 
	\end{equation}
	where $\mathbb{E}_\eta[\cdot]$ denotes the expectation with respect to the distribution of $\xi_k$ for a fixed $\eta$. In order to analyse the asymptotic properties of $\varepsilon_k$, we define the following initial problem, that is
	\begin{equation}
		\dot{\eta} = \bar{\varepsilon}(\eta), \quad \eta(0) = \varepsilon_0, 
	\end{equation}
	where $\varepsilon_0$ is the initial condition.
	Define $\varPhi = \bigl[\operatorname{Var}(\mathbb{E}[\langle x(k), \delta_N^\ell \rangle | \mathcal{F}_k^{Y(s)}])\bigr]_{s \in \mathcal{N}}$. Notice that the value of $\mathbb{E}[\langle x(k), \delta_N^\ell \rangle | \mathcal{F}_k^{Y(s)}]$ is less than or equal to 1, then it follows that $\operatorname{Var}(\mathbb{E}[\langle x(k), \delta_N^\ell \rangle | \mathcal{F}_k^{Y(s)}])$ is bounded. Furthermore, we can obtain that the expectation of $|\bar{\varepsilon}(\eta, \xi_k)|$ satisfies
	\begin{equation}
		\mathbb{E}\bigl[|\bar{\varepsilon}(\eta, \xi_k)|\bigr] < \infty, 
	\end{equation}
	and
	\begin{equation}
		\lim_{|\mathcal{N}| \to \infty} \sum_{s \in \mathcal{N}}\frac{1}{s^2}   \operatorname{Var}(\mathbb{E}[\langle x(k), \delta_N^\ell \rangle | \mathcal{F}_k^{Y(s)}]) < \infty. 
	\end{equation}
	Here, we use the strong law of large numbers to specify the mean vector field $\bar{\varepsilon}(\cdot)$ . Define the average
	\begin{equation}
		\bar{\varepsilon}_k(\eta) = \frac{1}{k} \sum_{t=1}^{k} \bar{\varepsilon}(\eta, \xi_t).
	\end{equation}
	
	\begin{proposition}
		If Assumption 2 holds, then there exists a finite $\bar{\varepsilon}(\eta)$ such that $\lim_{k \to \infty} \bar{\varepsilon}_k(\eta) = \bar{\varepsilon}(\eta)$ $\mathbb{P}$-a.s. is satisfied uniformly in $\eta$, where
		\begin{equation}
			\bar{\varepsilon}(\eta) = \Lambda\eta + \varGamma(\bar{e} - \tfrac{1}{|\mathcal{N}|}\mathbf{1}\mathbf{1}^\top \bar{e}),
		\end{equation}
		and $\bar{e} = [\bar{e}^s]_{s \in \mathcal{N}}$, in which $\bar{e}^s$ satisfies \[
		[\bar{e}^s(l)]_{l \in [1:N]} = \frac{\left[(H^s)^\top \mathcal{T}^s y^s(\infty)\right] \circ \left[FK \Phi_n [\bar{e}^s(l)]_{l \in [1:N]}\right]}{\left\langle (H^s)^\top \mathcal{T}^s y^s(\infty), FK \Phi_n [\bar{e}^s(l)]_{l \in [1:N]} \right\rangle}.
		\]
	\end{proposition}
	
	\begin{proof}
		See Appendix D
	\end{proof}
	
	\begin{definition}
		Let \(c\) be a positive constant. The origin is said to be globally asymptotically \(c\)-stable with respect to the initial value problem (25) if both of the following conditions hold:
		(1) For every \(\varepsilon > c\) there exists \(\delta > 0\) such that any trajectory \(\eta(t)\) of (25) with \(\|\eta(0)\| < \delta\) satisfies \(\|\eta(t)\| < \varepsilon\) for all \(t \geq 0\).
		(2) For every initial value \(\eta(0) \in \mathbb{R}^n\) there exists a finite time \(T \geq 0\) such that \(\|\eta(t)\| \leq c\) for all \(t \geq T\).
	\end{definition}
	
	\begin{proposition}
		Consider the initial value problem (25), and take the Lyapunov function candidate $V(\eta) = \frac{1}{2}\eta^\top\eta$. Define the compact set $\Omega_1 \triangleq \{\eta: V(\eta) \leq \frac{1}{2}c_1^2\}$, where the constant $c_1$ is given by
		\begin{equation}
			c_1 = \sqrt{(1+d_{\max})^2 - 2\lambda_2(L)} \cdot \|\bar{e} - \bar{e}_{\mathrm{avg}}\|_2 \cdot |\lambda_{\max}(\Lambda)|^{-1}, 
		\end{equation}
		with $\lambda_2$ and $\lambda_{\max}$ denoting the second-smallest and largest eigenvalue of the corresponding matrix, respectively.Then, the origin is globally asymptotically $c_1$-stable with respect to (25). Specifically, for every $\varepsilon > c_1$ there exists a $\delta > 0$ such that every trajectory of (25) with $\|\eta(0)\| < \delta$ satisfies $\|\eta(t)\| < \varepsilon$ for all $t \geq 0$.
	\end{proposition}
	
	\begin{proof}
		See Appendix E.
	\end{proof}
	Subsequently, we employ a stochastic stability approach to establish the \emph{recurrence} of the process $\{\varepsilon_k\}$, meaning that the error process $\{\varepsilon_k\}$ visits some compact set infinitely often with probability one. We define the filtration $\{\widetilde{\mathcal{F}}_k\}$ by $\widetilde{\mathcal{F}}_k \triangleq [\widehat{\mathcal{F}}_k^s]_{s \in \mathcal{N}}$, where each $\sigma$-algebra $\widehat{\mathcal{F}}_k^s$ is given by $\widehat{\mathcal{F}}_k^s = \sigma\bigl(\{\varepsilon_0^s\} \cup (\bigcup_{s' \in \mathcal{N}_s \cup \{s\}} \{y^{s'}(1), \ldots, y^{s'}(k)\})\bigr)$. Let $\mathbb{E}_k[\cdot]$ denote the conditional expectation operator with respect to $\widetilde{\mathcal{F}}_k$. The discount factor $c_k^i$ is defined as $c_k^i = (1-\rho)^{i-k+1}$ for $i \geq k$, and $c_k^i = 1$ for $i < k$. The discounted perturbation $\Delta_k(\eta): \mathbb{R}^n \to \mathbb{R}^n$ is defined by
	\begin{equation}
		\Delta_k(\eta) = \sum_{i=k}^{\infty} \rho c_{k+1}^i \mathbb{E}_k\bigl[\bar{\varepsilon}(\eta, \xi_{i+1}) - \bar{\varepsilon}(\eta)\bigr]. \label{eq:27}
	\end{equation}
	Observe that the supremum of the infinite sum is finite $\sup_k \sum_{i=k}^{\infty} \rho c_{k+1}^i < \infty$. Consequently, it follows that
	\begin{equation}
		\mathbb{E}_k \left[ \Delta_k(\eta) \right] = \sum_{i=k+1}^{\infty} \rho c_{k+1}^i \mathbb{E}_k\bigl[\bar{\varepsilon}(\eta, \xi_{i+1}) - \bar{\varepsilon}(\eta)\bigr] \quad \mathbb{P}\text{-a.s.}
	\end{equation}
	We define the perturbed stochastic Lyapunov function candidate as
	\begin{equation}
		V_k(\varepsilon_k) \triangleq V(\varepsilon_k) + \left. \nabla V(\eta) \right|_{\eta=\varepsilon_k} \Delta_k(\varepsilon_k), \label{eq:29}
	\end{equation}
	where $\nabla V(\cdot)$ denotes the gradient of $V(\cdot)$.
	
	\begin{assumption}
		The step size $\rho$ is assumed to be strictly positive and bounded above by $2(1+3d_{\max})^{-1}$, i.e., $0 < \rho < 2(1+3d_{\max})^{-1}$, where $d_{max} \triangleq \max_{i \in \mathcal{N}} \sum_{j \in \mathcal{N}_i} a_{ij}.$
	\end{assumption}
	
	\begin{theorem}
		Suppose Assumptions 1, 2, and 3 are satisfied. Let the real-valued Lyapunov function $V(\cdot)$ for the initial value problem (25) possesses bounded second mixed partial derivatives. Then, the perturbed stochastic Lyapunov function $V_k(\varepsilon_k)$ is an $\mathcal{F}_k$-supermartingale for the stopped process $\{\varepsilon_k\}$ up to the first hitting time of some compact set $\Omega_2=\{\eta : V(\eta) \le c_2\}$ with $c_2\in(0,\infty)$.
	\end{theorem}
	
	\begin{proof}
		See Appendix F.
	\end{proof}
	The subsequent theorem establishes the recurrence property of the error sequence $\{\varepsilon_k\}$.
	\begin{theorem}
		The perturbed stochastic Lyapunov function $V_k(\varepsilon_k)$ given by \eqref{eq:29} is a real-valued supermartingale adapted to the filtration $\widetilde{\mathcal{F}}_k$. Provided that the initial expectation $\mathbb{E}[V(\varepsilon_0)]$ is finite, there exists a compact set $L_p$ such that the sequence $\{\varepsilon_k\}$ visits $L_p$ infinitely often with probability at least $p$, for any $p \in (0,1]$.
	\end{theorem}
	
	\begin{proof}
		See Appendix G.
	\end{proof}
	In what follows, we establish the almost sure convergence of the error sequence $\{\varepsilon_k\}$ to a bounded invariant set. If the origin is globally asymptotically $c_1$-stable for the initial value problem (25) with some invariant level set, then $\{\varepsilon_k\}$ converges to the invariant set $\Omega_1$ almost surely.
	
	\begin{lemma}
		There exist nonnegative measurable functions $\alpha(\cdot)$ and $\beta_k(\cdot)$ of $\eta$ and $\xi$, respectively, such that $\alpha(\cdot)$ is bounded on every bounded subset of the $\eta$-space, and
		\[
		\|\bar{\varepsilon}(\eta, \xi) - \bar{\varepsilon}(\tilde{\eta}, \xi)\| \leq \alpha(\eta - \tilde{\eta})\beta_k(\xi),
		\]
		where $\alpha(\eta) \to 0$ as $\eta \to \mathbf{0}$, and $\beta_k(\cdot)$ satisfies
		\[
		\mathbb{P}\Bigl[\limsup_{l \to \infty} \sum_{k=l}^{\lfloor \rho^{-1}(t_\ell + \bar{\tau}) \rfloor} \rho \beta_k(\xi_k) < \infty\Bigr] = 1,
		\]
		for some $\bar{\tau} > 0$. 
	\end{lemma}
	
	\begin{proof}
		Applying the Gersgorin theorem to the negative definite matrix $\Lambda$, we obtain that its minimum eigenvalue satisfies $\lambda_{1}(\Lambda) \geq -(1+3d_{\max})$. Consequently, the following norm inequality holds:
		\[
		\|\Lambda(\eta - \tilde{\eta})\| \leq (1+3d_{\max})\|\eta - \tilde{\eta}\|.
		\]
		Here, we set $\beta_k(\xi) = 1+3d_{\max}$ for all $\xi$ and $\bar{\tau} > 0$. Furthermore, the function $\alpha(\eta) = \|\eta\|$ is bounded on every bounded subset in the space of $\eta$ and converges to zero as $\eta \to \mathbf{0}$.
	\end{proof}

	\begin{lemma}
		Suppose Assumption 2 holds. For any fixed \(\eta \in \mathbb{R}^{rN}\), define the cumulative perturbation process
		\[
		\overline{\varepsilon}_{\eta}(t)=\sum_{i=0}^{\left\lfloor t \rho^{-1}\right\rfloor-1} \rho\left[\overline{\varepsilon}\left(\eta, \xi_{i+1}\right)-\overline{\varepsilon}(\eta)\right].
		\]
		Then the rate of change of \(\overline{\varepsilon}_{\eta}(t)\) tends to zero almost surely as \(t \to \infty\), i.e.,
		\[
		\lim_{t \to \infty} \frac{1}{t} \overline{\varepsilon}_{\eta}(t) = 0, \quad \mathbb{P}\text{-a.s.}
		\]
	\end{lemma}
	
	\begin{proof}
		We first decompose the one-step perturbation term. By the definition of \(\overline{\varepsilon}(\eta, \xi_{k+1})\) in (23) and the stationary mean vector field \(\overline{\varepsilon}(\eta)\) in Proposition 1, we have
		\[
		\overline{\varepsilon}(\eta, \xi_{k+1}) - \overline{\varepsilon}(\eta)
		= \Gamma\left( e_{k+1} - \bar{e} - \frac{1}{|\mathcal{N}|} \mathbf{1}\mathbf{1}^\top (e_k - \bar{e}) \right)
		- \frac{1}{|\mathcal{N}|\rho} \mathbf{1}\mathbf{1}^\top \left( e_{k+1} - \bar{e} - (e_k - \bar{e}) \right),
		\]
		where \(\bar{e} = \lim_{k\to\infty} \mathbb{E}[e_k]\) denotes the stationary mean of the local estimate sequence, and \(\mathbf{1}\) is the all-one vector of compatible dimension.
		Let \(d_k = e_k - \bar{e}\) stand for the deviation of the local estimate from its stationary mean. The right-hand side is a linear transformation of \(d_k\) and \(d_{k+1}\). We denote the resulting zero-mean random sequence by
		\[
		\phi_k \triangleq \overline{\varepsilon}(\eta, \xi_{k+1}) - \overline{\varepsilon}(\eta), \quad k \in \mathbb{Z}^+.
		\]
		Under Assumption 2, the state process \(\{x(k)\}\) is geometrically ergodic under state feedback control. The joint observation process \(\{(x(k), y^s(k))\}\) for each sensor \(s\) inherits geometric ergodicity and admits a unique stationary distribution. Since each entry of \(e_k^s(l)\) is a bounded measurable functional of the observation filtration \(\mathcal{F}_k^{Y(s)}\), the sequence \(\{e_k\}\) is stationary and ergodic. It follows that \(\{\xi_k = (e_k, e_{k-1})\}\) and hence \(\{\phi_k\}\) are stationary ergodic sequences.
		By definition of the mean vector field, the expectation of \(\phi_k\) satisfies
		\[
		\mathbb{E}[\phi_k] = \mathbb{E}\left[\overline{\varepsilon}(\eta, \xi_{k+1})\right] - \overline{\varepsilon}(\eta) = 0.
		\]
		Moreover, each entry of \(e_k\) takes values in \([0,1]\), so \(\{\phi_k\}\) is uniformly bounded and thus integrable.
		Applying the strong law of large numbers for stationary ergodic processes to the zero-mean integrable sequence \(\{\phi_k\}\) yields
		\[
		\lim_{m \to \infty} \frac{1}{m} \sum_{i=0}^{m-1} \phi_i = 0, \quad \mathbb{P}\text{-a.s.}
		\]
		We now connect the discrete sum to the continuous-time process \(\overline{\varepsilon}_\eta(t)\). For any \(t>0\), let \(m = \lfloor t \rho^{-1} \rfloor\). By construction
		\[
		\overline{\varepsilon}_\eta(t) = \rho \sum_{i=0}^{m-1} \phi_i.
		\]
		Note that \(m \rho \leq t < (m+1)\rho\), so \(t / m \to \rho\) as \(t \to \infty\). Taking the time-normalized rate gives
		\[
		\frac{1}{t} \overline{\varepsilon}_\eta(t)
		= \frac{\rho m}{t} \cdot \frac{1}{m} \sum_{i=0}^{m-1} \phi_i.
		\]
		As \(t \to \infty\), the factor \(\rho m / t\) converges to 1, and the sample average converges to 0 almost surely. Therefore
		\[
		\lim_{t \to \infty} \frac{1}{t} \overline{\varepsilon}_{\eta}(t) = 0, \quad \mathbb{P}\text{-a.s.}
		\]
		This completes the proof.
	\end{proof}
	
	\begin{theorem}
		There exists a null set $\Psi$ such that for all $\omega \notin \Psi$, the collection of functions $\{\mathcal{I}_n(\omega, \cdot), n < \infty\}$ is equicontinuous. Let $\mathcal{I}(\omega, \cdot)$ denote the limit of some convergent subsequence $\{\mathcal{I}_{n'}(\omega, \cdot)\}$. It follows that, for $\mathbb{P}$-almost every $\omega \in \Omega$, the limit $\mathcal{I}(\omega, \cdot)$ is a trajectory of the initial value problem (25) contained within the bounded invariant set $\Omega_1$, and the error sequence $\{\varepsilon_k\}$ converges to $\Omega_1$. Furthermore, since the origin is globally asymptotically $c_1$-stable for the initial value problem (25) with the invariant level set $\Omega_1$, the sequence $\{\varepsilon_k\}$ satisfies
		\[
		\limsup_{k\to\infty}\|\varepsilon_k\| \le c_1, \quad \mathbb{P}\text{-a.s.}
		\]
	\end{theorem}
	
	\begin{proof}
		The Lemma 5 establishes that the asymptotic rate of change of the cumulative perturbation tends to zero. Coupled with the globally asymptotically \(c_1\)-stability of the origin for the mean-field initial value problem proved in Proposition 2, all hypotheses of Theorem 7.1 and Theorem 1.1 in Chapter 6 of \cite{37} hold. The almost-sure convergence of the error sequence \(\{\varepsilon_k\}\) to the bounded invariant set \(\Omega_1\) and the corresponding limsup bound follow directly.
	\end{proof} 
	
	\begin{theorem}
		For state $\delta_N^l$ with $l \in [1:N]$, let the global state conditional expectation average be $\theta_k^*$ and $\tilde{l} = \arg\max_{l\in[1:N]}\theta_k^*(l)$. Consider an arbitrary sensor subset $\mathcal{S} \subseteq \mathcal{N}$ with $|\mathcal{S}| \ge 1$, and define the fused estimate $\bar{\theta}_k^{\mathcal{S}} \triangleq \frac{1}{|\mathcal{S}|}\sum_{s \in \mathcal{S}} \hat{\theta}_k^s$.
		Define the \emph{admissible candidate set}
		\begin{equation}
			\label{eq:candidate_set}
			\mathcal{C}_k \triangleq \left\{ l \in [1:N] \,\middle|\, \bar{\theta}_k^{\mathcal{S}}(l) + \frac{c_1(l)}{\sqrt{|\mathcal{S}|}} \ge \max_{m \in [1:N]} \left( \bar{\theta}_k^{\mathcal{S}}(m) - \frac{c_1(m)}{\sqrt{|\mathcal{S}|}} \right) \right\},
		\end{equation}
		where $c_1(l)$ is given by (30). Then the following hold.
		\begin{enumerate}
			\item $\tilde{l} \in \mathcal{C}_k$ holds $\mathbb{P}$-a.s.
			\item If the \emph{separation condition}
			\begin{equation}
				\label{eq:separation}
				\bar{\theta}_k^{\mathcal{S}}(\tilde{l}) - \bar{\theta}_k^{\mathcal{S}}(l) > \frac{1}{\sqrt{|\mathcal{S}|}}\bigl(c_1(\tilde{l}) + c_1(l)\bigr)
			\end{equation}
			holds for all $l \in [1:N] \setminus \{\tilde{l}\}$, then $\mathcal{C}_k = \{\tilde{l}\}$. Moreover, $l^* = \arg\min_{l \in [1:N]} \|\bar{\theta}_k^{\mathcal{S}} - \delta_N^l\|_2$ is unique and satisfies $l^* = \tilde{l}$.
			\item If \eqref{eq:separation} is violated for some $l \neq \tilde{l}$, then $|\mathcal{C}_k| \ge 2$, and for every $l \in \mathcal{C}_k \setminus \{\tilde{l}\}$,
			\begin{equation}
				\label{eq:gap_bound}
				\theta_k^*(\tilde{l}) - \theta_k^*(l) \le \frac{c_1(\tilde{l}) + c_1(l)}{\sqrt{|\mathcal{S}|}} + \bigl[\bar{\theta}_k^{\mathcal{S}}(\tilde{l}) - \bar{\theta}_k^{\mathcal{S}}(l)\bigr]^+.
			\end{equation}
		\end{enumerate}
		Consequently, the optimal state estimate is $\bar{\theta}_k^* = \{\delta_N^l : l \in \mathcal{C}_k\}$, and node $i$'s estimate is $X_i^*(k) = D_r^{[2,2^{n-1}]} W_{[2,2^{i-1}]} \delta_N^{l^*}$ when $|\mathcal{C}_k|=1$, or $\{D_r^{[2,2^{n-1}]} W_{[2,2^{i-1}]} \delta_N^l : l \in \mathcal{C}_k\}$ otherwise.
	\end{theorem}
	
	\begin{proof}
		See Appendix H.
	\end{proof}
	
	\begin{remark}
		It is worth noting that when the separation condition $(34)$ is violated due to an insufficiently large probability gap, such a condition can be reconfigured by optimizing the communication network topology. From a theoretical perspective, the consensus error bound is quantitatively governed by the algebraic connectivity of the communication graph $\mathcal{G}$, i.e., the second smallest eigenvalue of the Laplacian matrix $\lambda_2$. By strategically adding critical communication links or increasing node degrees, one can enlarge $\lambda_2$ and accelerate the information fusion rate across the network. Mathematically, this effectively compresses the consensus estimation error bound, forcing it to shrink below the existing probability gap. 
	\end{remark}
	
	\begin{remark}
		In Theorem 6, the subset $\mathcal{S} \subseteq \mathcal{N}$ is chosen by an observer or local fusion unit according to application requirements. The observer may utilize a single sensor, i.e., $|\mathcal{S}|=1$ or fuse estimates from multiple sensors, i.e., $|\mathcal{S}|\ge 1$, up to the entire network. This selection can be adjusted dynamically, and Theorem 6 provides a unified decision framework for all such cases.
	\end{remark}
	
	\section{Simulation}
	This section presents comprehensive numerical simulation results to validate the effectiveness and theoretical properties of the proposed distributed stochastic average consensus filtering framework for BCNs. We consider a disturbed BCN with $n=3$ nodes, corresponding to a state space of size $N=8$. The dynamic of the considered BCN is described by the expected structure matrices. The overall state transition structure matrix of the BCN is constructed as
	$F = [F_1\ F_2\ F_3\ F_4]$,
	where each block of $F$ is given by
	\[
	F_1 = \begin{bmatrix}
		0.690 & 0.471 & 0.118 & 0.027 & 0.185 & 0.520 & 0.352 & 0.241\\
		0.140 & 0.066 & 0.285 & 0.318 & 0.004 & 0.160 & 0.000 & 0.052\\
		0.070 & 0.121 & 0.122 & 0.000 & 0.025 & 0.100 & 0.097 & 0.004\\
		0.020 & 0.008 & 0.096 & 0.036 & 0.006 & 0.070 & 0.367 & 0.356\\
		0.020 & 0.031 & 0.015 & 0.368 & 0.051 & 0.050 & 0.004 & 0.034\\
		0.020 & 0.030 & 0.101 & 0.124 & 0.244 & 0.040 & 0.159 & 0.224\\
		0.040 & 0.193 & 0.258 & 0.118 & 0.054 & 0.030 & 0.014 & 0.090\\
		0.000 & 0.080 & 0.003 & 0.008 & 0.432 & 0.030 & 0.007 & 0.000
	\end{bmatrix}
	\],
	
	\[
	F_2 = \begin{bmatrix}
		0.511 & 0.870 & 0.129 & 0.500 & 0.000 & 0.223 & 0.397 & 0.329\\
		0.000 & 0.130 & 0.006 & 0.120 & 0.191 & 0.015 & 0.071 & 0.010\\
		0.001 & 0.000 & 0.025 & 0.100 & 0.140 & 0.086 & 0.000 & 0.031\\
		0.000 & 0.000 & 0.401 & 0.100 & 0.015 & 0.252 & 0.073 & 0.015\\
		0.294 & 0.000 & 0.002 & 0.060 & 0.319 & 0.346 & 0.084 & 0.004\\
		0.007 & 0.000 & 0.008 & 0.050 & 0.098 & 0.001 & 0.000 & 0.003\\
		0.186 & 0.000 & 0.012 & 0.040 & 0.077 & 0.017 & 0.029 & 0.476\\
		0.000 & 0.000 & 0.417 & 0.030 & 0.159 & 0.060 & 0.345 & 0.132
	\end{bmatrix}
	\],
	
	\[
	F_3 = \begin{bmatrix}
		0.002 & 0.283 & 0.007 & 0.388 & 0.550 & 0.010 & 0.045 & 0.460\\
		0.006 & 0.039 & 0.002 & 0.059 & 0.100 & 0.182 & 0.028 & 0.180\\
		0.001 & 0.379 & 0.586 & 0.001 & 0.080 & 0.087 & 0.148 & 0.110\\
		0.013 & 0.165 & 0.399 & 0.255 & 0.070 & 0.043 & 0.017 & 0.090\\
		0.880 & 0.021 & 0.000 & 0.037 & 0.070 & 0.482 & 0.251 & 0.050\\
		0.000 & 0.066 & 0.000 & 0.033 & 0.050 & 0.100 & 0.006 & 0.040\\
		0.000 & 0.012 & 0.000 & 0.085 & 0.040 & 0.069 & 0.050 & 0.040\\
		0.099 & 0.034 & 0.005 & 0.142 & 0.040 & 0.026 & 0.455 & 0.030
	\end{bmatrix}
	\],
	
	\[
	F_4 = \begin{bmatrix}
		0.090 & 0.047 & 0.600 & 0.563 & 0.043 & 0.035 & 0.490 & 0.000\\
		0.030 & 0.041 & 0.100 & 0.020 & 0.000 & 0.103 & 0.150 & 0.114\\
		0.251 & 0.031 & 0.080 & 0.010 & 0.002 & 0.000 & 0.120 & 0.034\\
		0.012 & 0.399 & 0.060 & 0.104 & 0.045 & 0.108 & 0.080 & 0.005\\
		0.009 & 0.046 & 0.050 & 0.171 & 0.000 & 0.615 & 0.050 & 0.093\\
		0.010 & 0.078 & 0.050 & 0.061 & 0.023 & 0.093 & 0.040 & 0.157\\
		0.596 & 0.036 & 0.030 & 0.003 & 0.407 & 0.013 & 0.040 & 0.048\\
		0.002 & 0.322 & 0.030 & 0.068 & 0.479 & 0.033 & 0.030 & 0.549
	\end{bmatrix}
	\] The state feedback control gain $K \in \mathcal{L}_{4\times 8}$ is designed as
	\[
	K = \begin{bmatrix}
		1 & 0 & 0 & 0 & 0 & 1 & 0 & 0\\
		0 & 1 & 0 & 1 & 0 & 0 & 0 & 0\\
		0 & 0 & 0 & 0 & 1 & 0 & 0 & 1\\
		0 & 0 & 1 & 0 & 0 & 0 & 1 & 0
	\end{bmatrix}
	\] Under the designed state feedback control law $u(k)=Kx(k)$, the state transition matrix $Q=FK{\Phi_3}$ of the BCN is derived as
	\[
	Q = \begin{bmatrix}
		0.69 & 0.87 & 0.60 & 0.50 & 0.55 & 0.52 & 0.49 & 0.46\\
		0.14 & 0.13 & 0.10 & 0.12 & 0.10 & 0.16 & 0.15 & 0.18\\
		0.07 & 0.00 & 0.08 & 0.10 & 0.08 & 0.10 & 0.12 & 0.11\\
		0.02 & 0.00 & 0.06 & 0.10 & 0.07 & 0.07 & 0.08 & 0.09\\
		0.02 & 0.00 & 0.05 & 0.06 & 0.07 & 0.05 & 0.05 & 0.05\\
		0.02 & 0.00 & 0.05 & 0.05 & 0.05 & 0.04 & 0.04 & 0.04\\
		0.04 & 0.00 & 0.03 & 0.04 & 0.04 & 0.03 & 0.04 & 0.04\\
		0.00 & 0.00 & 0.03 & 0.03 & 0.04 & 0.03 & 0.03 & 0.03
	\end{bmatrix}
	\] The communication graph of the sensor network is specified by the weighted adjacency matrix
	\[
	A = \begin{bmatrix}
		0 & 0 & 0 & 0.6 & 0\\
		0 & 0 & 0.2 & 0 & 0.6\\
		0 & 0.2 & 0 & 0.6 & 0.5\\
		0.6 & 0 & 0.6 & 0 & 0\\
		0 & 0.6 & 0.5 & 0 & 0
	\end{bmatrix}
	\]
	
	\begin{figure}
	    \centering
	    \includegraphics[width=0.5\linewidth]{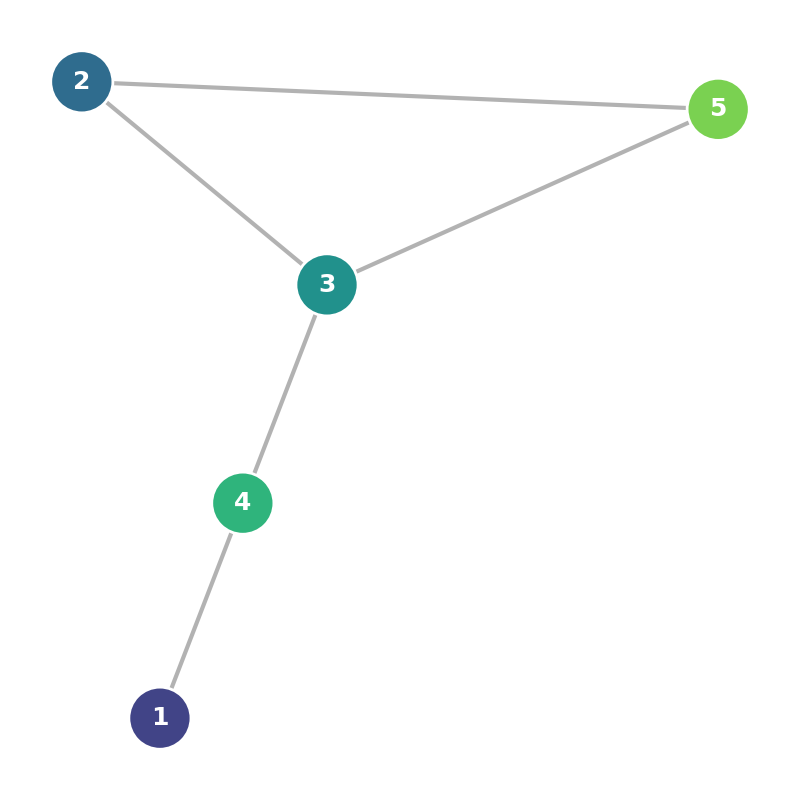}
	    \caption{Sensor Network}
	    \label{fig:1}
	\end{figure}
	
	Each sensor node $s$ is assigned an observation matrix $H^s$, where each observation matrix is given by
	\[
	H^1 = \begin{bmatrix}
		1 & 0 & 0 & 0 & 0 & 0 & 0 & 0.4 \\
		0 & 0.8 & 0 & 0.1 & 0 & 0 & 0 & 0 \\
		0 & 0 & 0 & 0 & 0 & 1 & 0 & 0 \\
		0 & 0 & 0 & 0.7 & 0 & 0 & 0 & 0 \\
		0 & 0 & 1 & 0 & 0 & 0 & 0 & 0 \\
		0 & 0 & 0 & 0 & 0 & 0 & 0 & 0.6 \\
		0 & 0 & 0 & 0 & 0 & 0 & 1 & 0 \\
		0 & 0.2 & 0 & 0.2 & 1 & 0 & 0 & 0
	\end{bmatrix}
	\],
	\[
	H^2 = \begin{bmatrix}
		0 & 0 & 0 & 0 & 0.1 & 0 & 0 & 1 \\
		0 & 1 & 0 & 0 & 0 & 0 & 0 & 0 \\
		0 & 0 & 0 & 0 & 0 & 1 & 0 & 0 \\
		0 & 0 & 0 & 1 & 0 & 0 & 0 & 0 \\
		1 & 0 & 0 & 0 & 0 & 0 & 0 & 0 \\
		0 & 0 & 0.4 & 0 & 0.9 & 0 & 0 & 0 \\
		0 & 0 & 0 & 0 & 0 & 0 & 0.9 & 0 \\
		0 & 0 & 0.6 & 0 & 0 & 0 & 0.1 & 0
	\end{bmatrix}
	\],
	\[
	H^3 = \begin{bmatrix}
		0 & 0 & 0 & 0 & 1 & 0 & 0 & 0 \\
		0 & 0.8 & 0 & 0 & 0 & 0.4 & 0 & 0 \\
		0 & 0 & 0 & 1 & 0 & 0 & 0 & 0 \\
		0 & 0.2 & 0 & 0 & 0 & 0 & 0 & 1 \\
		0 & 0 & 1 & 0 & 0 & 0 & 0 & 0 \\
		1 & 0 & 0 & 0 & 0 & 0 & 0 & 0 \\
		0 & 0 & 0 & 0 & 0 & 0 & 1 & 0 \\
		0 & 0 & 0 & 0 & 0 & 0.6 & 0 & 0
	\end{bmatrix}
	\],
	\[
	H^4 = \begin{bmatrix}
		0 & 0.1 & 0 & 0 & 0 & 0 & 1 & 0 \\
		1 & 0 & 0 & 0 & 0 & 0 & 0 & 0 \\
		0 & 0 & 0 & 0.2 & 0 & 0 & 0 & 1 \\
		0 & 0 & 0 & 0.7 & 0.1 & 0 & 0 & 0 \\
		0 & 0 & 0 & 0 & 0 & 1 & 0 & 0 \\
		0 & 0.8 & 0.2 & 0 & 0 & 0 & 0 & 0 \\
		0 & 0.1 & 0 & 0 & 0.9 & 0 & 0 & 0 \\
		0 & 0 & 0.8 & 0.1 & 0 & 0 & 0 & 0
	\end{bmatrix}
	\],
	\[
	H^5 = \begin{bmatrix}
		0.1 & 0 & 0 & 0 & 0.2 & 0 & 0 & 0.85 \\
		0 & 0 & 0 & 0 & 0 & 1 & 0 & 0.05 \\
		0 & 0 & 0 & 1 & 0.1 & 0 & 0 & 0.1 \\
		0 & 1 & 0 & 0 & 0 & 0 & 0 & 0 \\
		0 & 0 & 0.8 & 0 & 0 & 0 & 0 & 0 \\
		0 & 0 & 0.1 & 0 & 0 & 0 & 0.8 & 0 \\
		0 & 0 & 0.1 & 0 & 0.7 & 0 & 0.2 & 0 \\
		0.9 & 0 & 0 & 0 & 0 & 0 & 0 & 0
	\end{bmatrix}
	\].
	
	\begin{figure}
	    \centering
	    \includegraphics[width=0.5\linewidth]{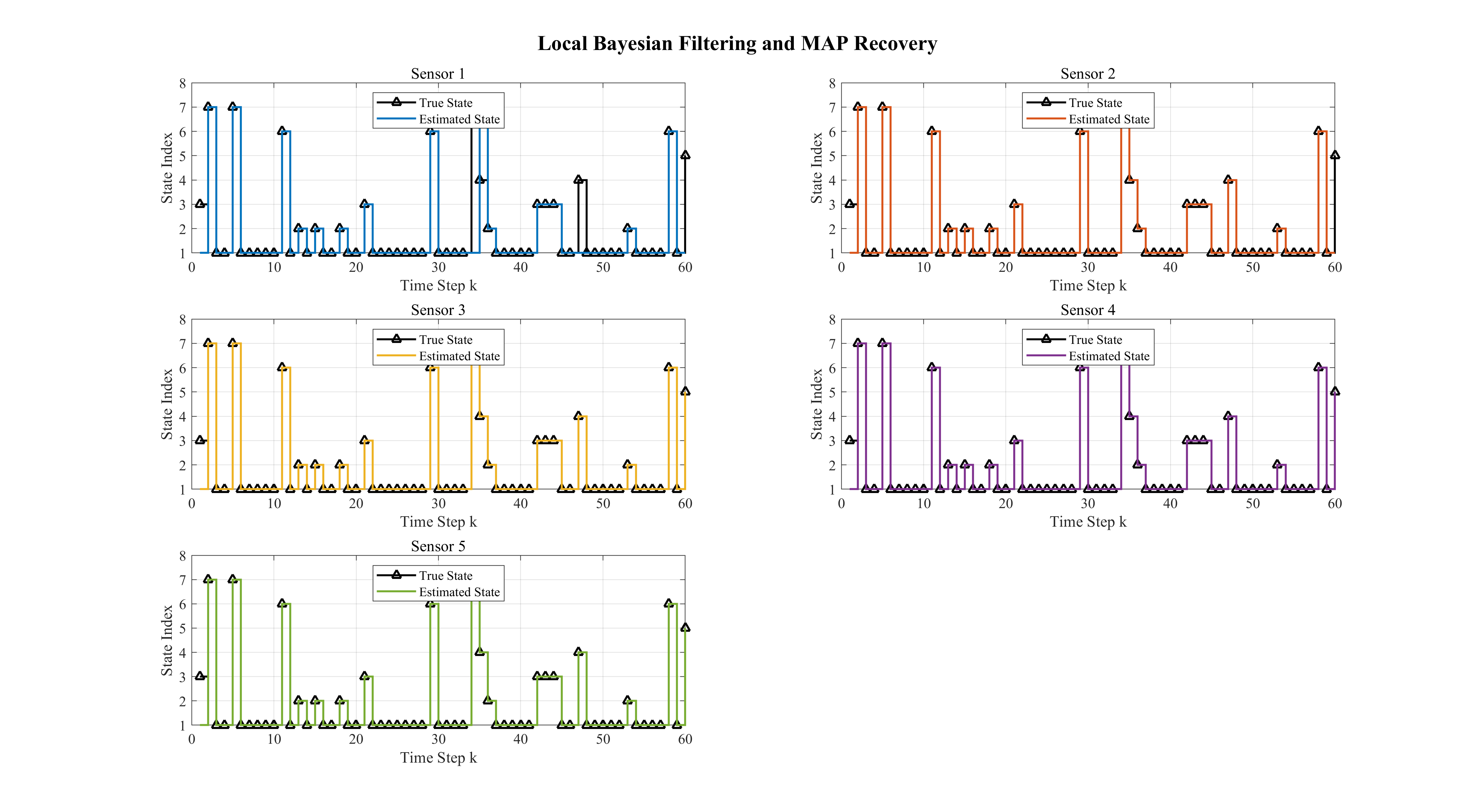}
	    \caption{Local Filtering}
	    \label{fig:plaaceholder}
	\end{figure}

	Figure 2 depicts the distributed state tracking trajectories of the five sensor nodes over 60 time steps. All nodes are initialized with a uniform prior distribution $\hat{\theta}_0^s=\mathbf{1}_8/8$. The maximum a posteriori (MAP) state estimate at each node is derived by selecting the index corresponding to the maximum entry of the posterior probability vector. The state recovery accuracy is then calculated by comparing the MAP estimate with the true system state. All sensors exhibit non-negligible estimation errors in the initial transient phase, and converge to steady state tracking after approximately 10 time steps. Sensor 3, which is located at the topological center of the communication graph, achieves the fastest convergence rate and the highest accuracy, benefiting from its ability to fuse information from more neighboring nodes. When system state transitions occur, all nodes can re-lock the true system state within 3 time steps.

	\begin{figure}
	    \centering
	    \includegraphics[width=0.5\linewidth]{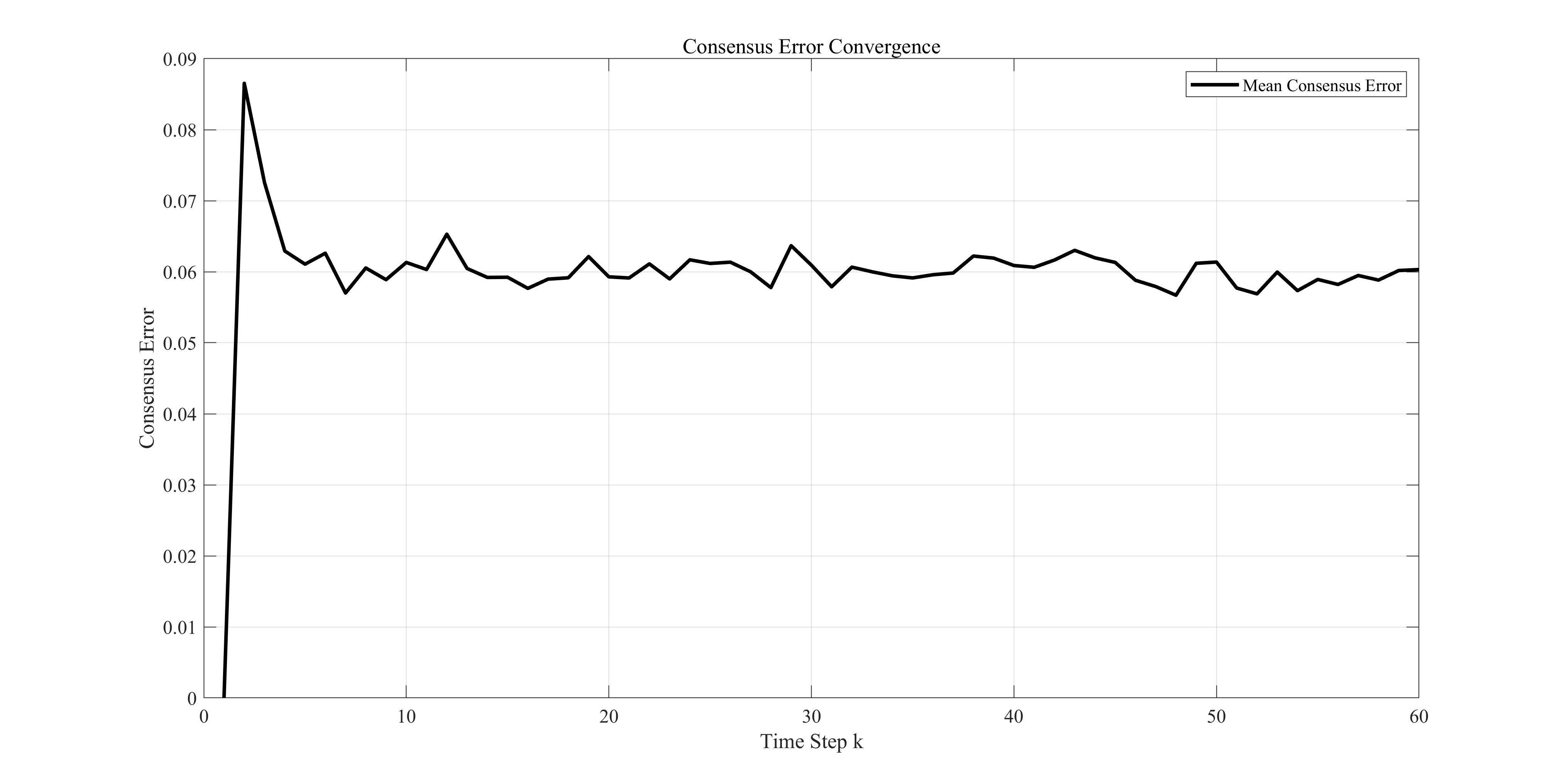}
	    \caption{Consensus Error Convergence}
	    \label{fig:2}
	\end{figure}
	
	Figure 3 presents the convergence of the average consensus error. The consensus errors converge to a steady state oscillation regime. Further analysis reveals that this stratification is mainly determined by the column concentration of $H^s$, rather than the topological connectivity of the communication graph alone. The results show that the estimation errors  remain bounded, but do not asymptotically converge to zero because of the observation noise.
	
	\begin{figure}
	    \centering
	    \includegraphics[width=0.5\linewidth]{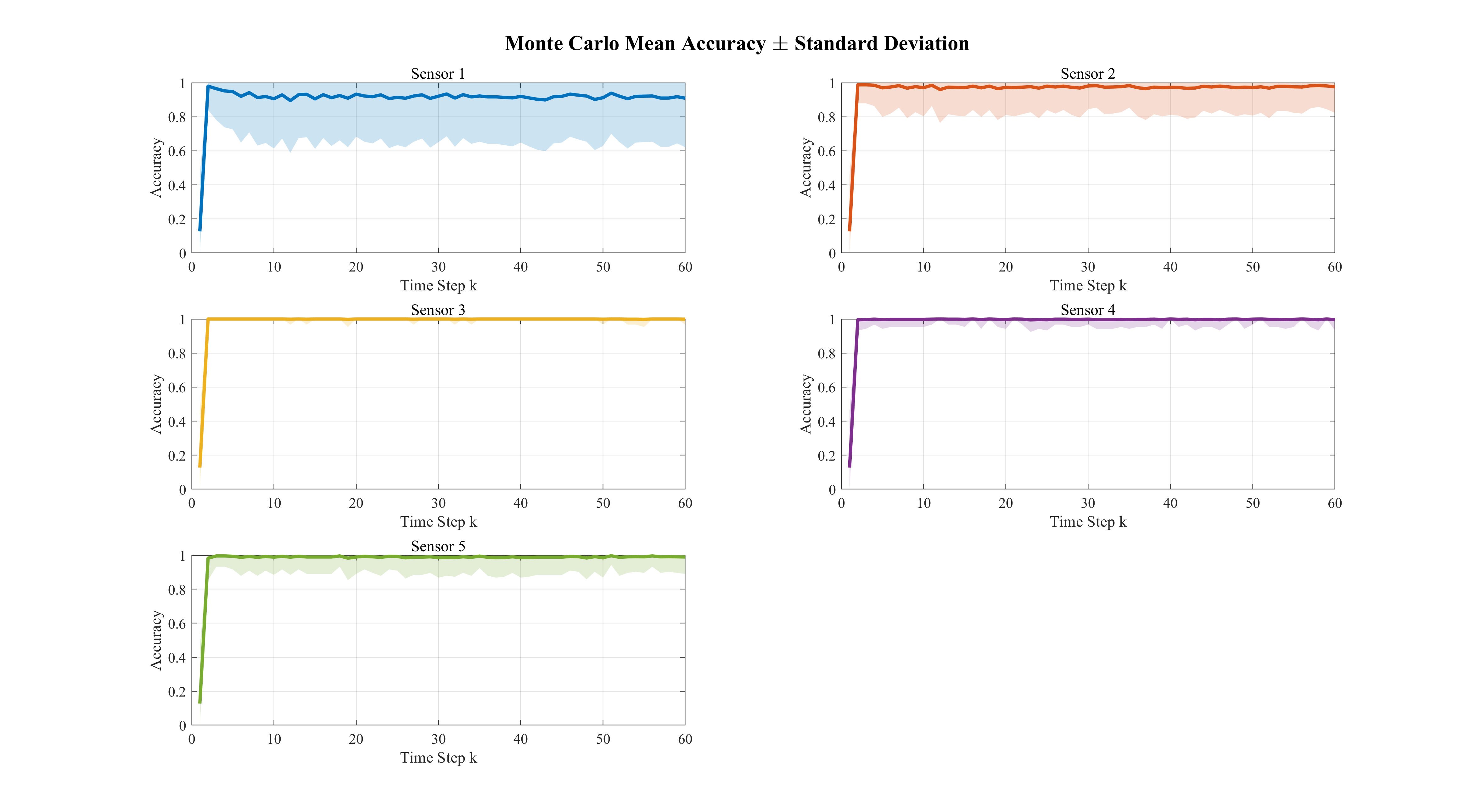}
	    \caption{Monte Carlo Accuracy}
	    \label{fig:3}
	\end{figure}
	
	Figure 4 and 5 shows the statistical evolution of the MAP state recovery accuracy for each sensor node. All sensor nodes start from an initial accuracy of $\approx 0.125$, and rapidly converge to the steady state  within the first three time steps. The sensor 1 exhibits the lowest accuracy and the largest performance fluctuations, which is attributed to its topological isolation and limited observation.
	
    \begin{figure}
        \centering
        \includegraphics[width=0.5\linewidth]{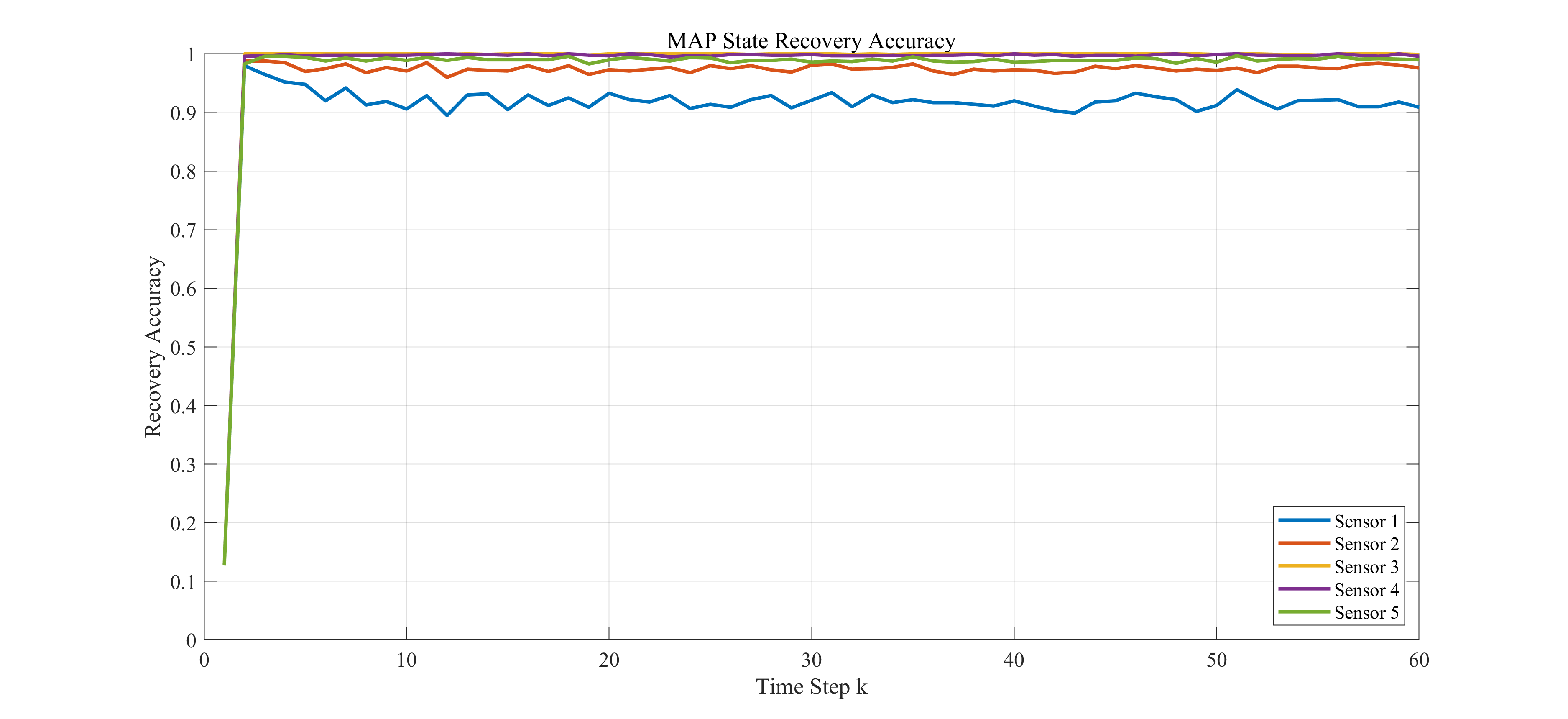}
        \caption{MAP State Recovery Accuracy}
        \label{fig:4}
    \end{figure}
    
	\begin{figure}
	    \centering
	    \includegraphics[width=0.5\linewidth]{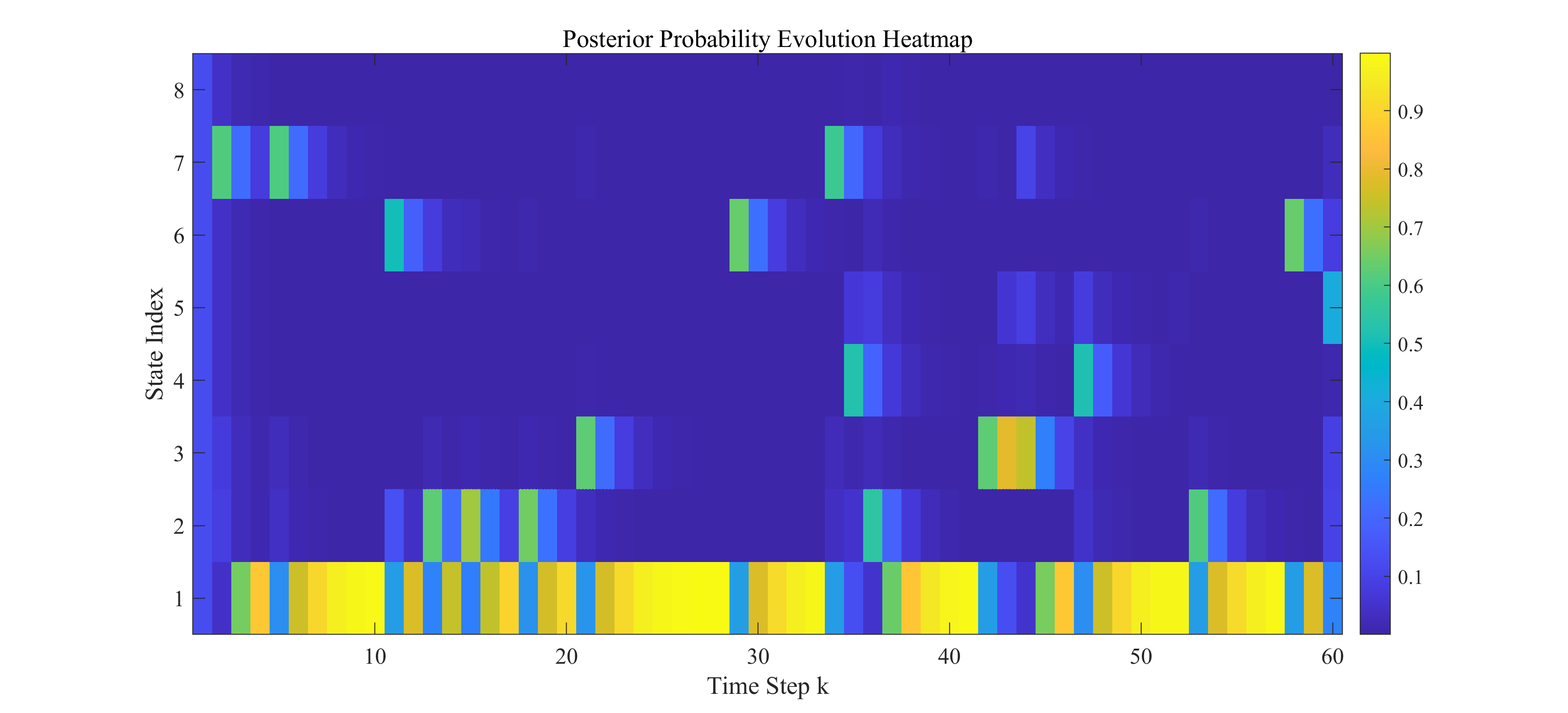}
	    \caption{Heatmap}
	    \label{fig:5}
	\end{figure}
	
	Figure 6 provides a heatmap of the global fused posterior probability $\frac{1}{|\mathcal{N}|}\sum_{s\in\mathcal{N}}\hat{\theta}_k^s$. The color intensity encodes the probability value. The heatmap shows that at most time instants the global posterior presents a sharp unimodal distribution, where the separation condition is strictly satisfied. Under this condition, theorem 6 guarantees that the MAP estimate is unique and consistent with the true system state. 
	
	\begin{figure}
	    \centering
	    \includegraphics[width=0.5\linewidth]{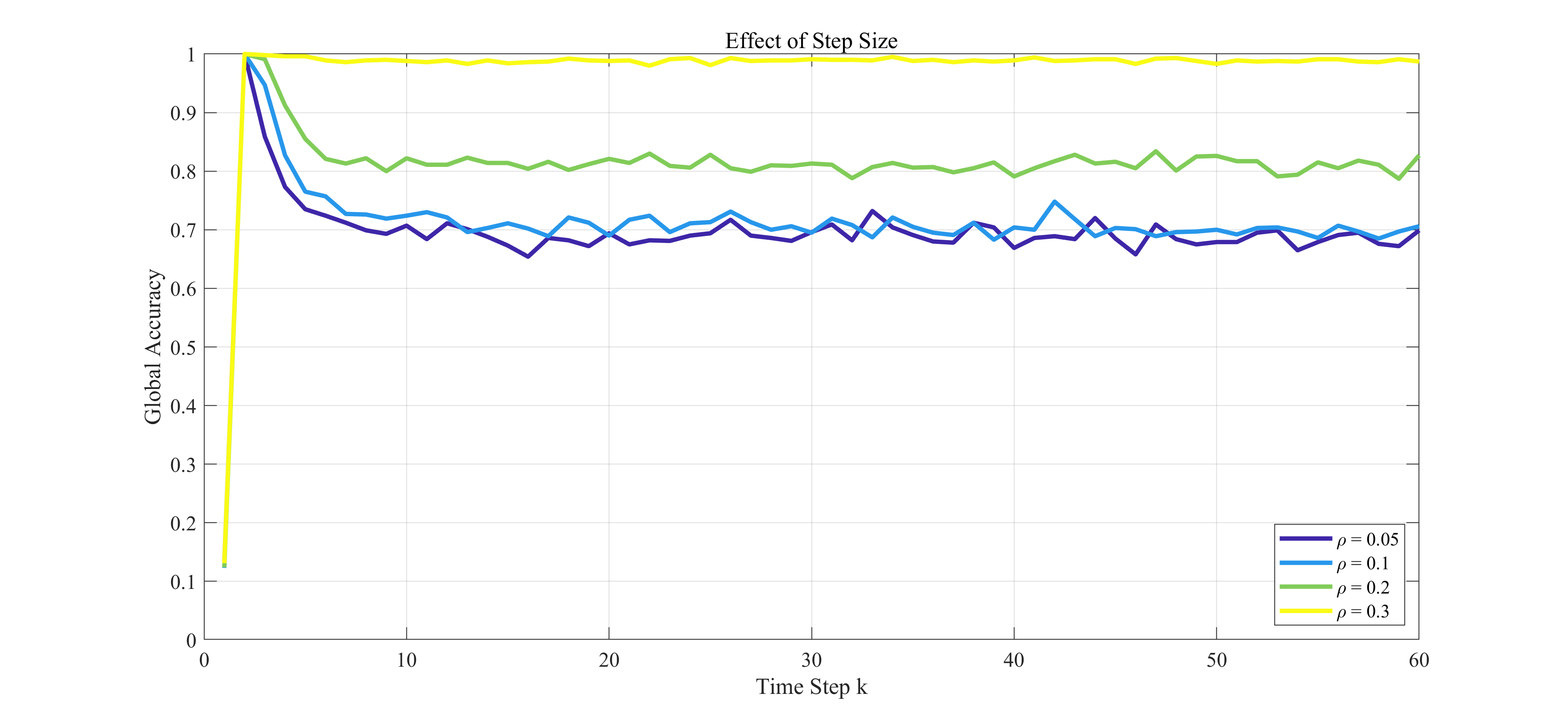}
	    \caption{Stepsize}
	    \label{fig:6}
	\end{figure}
	
	Figure 7 investigates the influence of the consensus step size $\rho$. All step sizes satisfy the step size constraint in Assumption 3. The step size has a significant impact on the steady state estimation performance. $\rho=0.3$ achieves the highest and most stable global estimation accuracy, while $\rho=0.05$ leads to the worst estimation performance. This trend indicates that a larger step size is more conducive to the rapid tracking of discrete state transitions in BCNs, while an excessively small step size places excessive weight on historical estimates, leading to a sluggish dynamic response.

	\section{Conclusion}
	This paper investigates the distributed multi-sensor fusion state estimation and stochastic average consensus filtering for disturbed Boolean control networks. By combining probability measure transformation, semi-tensor product and stochastic approximation, a recursive local fusion filtering approach and a distributed stochastic average consensus filter are proposed. The almost sure convergence of the algorithm is rigorously proved by means of the martingale convergence theorem and perturbed stochastic Lyapunov functions. Under connected topologies and appropriate step-size conditions, the estimation error of each sensor node is uniformly ultimately bounded by an explicitly quantified constant, with the bound governed by the network algebraic connectivity and sensor heterogeneity. The radius of this bound can be systematically reduced by enhancing the network algebraic connectivity. Future work will extend the proposed method to time-varying topologies, communication delays, packet dropouts and the co-design of state estimation and optimal control.
	
	\appendix
	
	\section{Proof of Lemma 1}
	
	\begin{proof}
		We first compute the conditional expectation with respect to the filtration $\mathcal{F}_k^X \cup \mathcal{F}_{k-1}^U \cup \mathcal{F}_{k-1}^{Y(s)}$
		\begin{align}
			&\mathbb{E}^s\bigl[\lambda_k^s | \mathcal{F}_k^X \cup \mathcal{F}_{k-1}^U \cup \mathcal{F}_{k-1}^{Y(s)}\bigr] \notag\\
			&= \mathbb{E}^s\Bigl[M_s \sum_{j=1}^{M_s} \langle H^s x(k), \delta_{M_s}^j \rangle \langle y^s(k), \delta_{M_s}^j \rangle \Bigm| \mathcal{F}_k^X \cup \mathcal{F}_{k-1}^U \cup \mathcal{F}_{k-1}^{Y(s)}\Bigr] \notag\\
			&= M_s \sum_{j=1}^{M_s} \langle H^s x(k), \delta_{M_s}^j \rangle \mathbb{E}^s\bigl[\langle y^s(k), \delta_{M_s}^j \rangle | \mathcal{F}_k^X \cup \mathcal{F}_{k-1}^U \cup \mathcal{F}_{k-1}^{Y(s)}\bigr] \notag\\
			&= \sum_{j=1}^{M_s} \langle H^s x(k), \delta_{M_s}^j \rangle \notag\\
			&= 1.
		\end{align}
		Similarly, we compute the conditional expectation with respect to the coarser filtration $\mathcal{F}_{k-1}^X \cup \mathcal{F}_{k-1}^U \cup \mathcal{F}_{k-1}^{Y(s)}$ by the law of iterated expectations, that is
		\begin{align}
			&\mathbb{E}^s\bigl[\lambda_k^s | \mathcal{F}_{k-1}^X \cup \mathcal{F}_{k-1}^U \cup \mathcal{F}_{k-1}^{Y(s)}\bigr] \notag\\
			&= \mathbb{E}^s\bigl[\mathbb{E}^s[\lambda_k^s | \mathcal{F}_k^X \cup \mathcal{F}_{k-1}^U \cup \mathcal{F}_{k-1}^{Y(s)}] | \mathcal{F}_{k-1}^X \cup \mathcal{F}_{k-1}^U \cup \mathcal{F}_{k-1}^{Y(s)}\bigr] \notag\\
			&= 1.
		\end{align}
	\end{proof}
	
	\section{Proof of Theorem 1}
	
	\begin{proof}
		Observe that $\Lambda_k^s$ is measurable with respect to $\mathcal{F}_k^X \cup \mathcal{F}_k^U \cup \mathcal{F}_k^{Y(s)}$. It follows that
		\begin{align}
			&\mathbb{E}^s\bigl[\Lambda_k^s | \mathcal{F}_{k-1}^X \cup \mathcal{F}_{k-1}^U \cup \mathcal{F}_{k-1}^{Y(s)}\bigr] \notag\\
			&= \Lambda_{k-1}^s \mathbb{E}^s\bigl[\lambda_k^s | \mathcal{F}_{k-1}^X \cup \mathcal{F}_{k-1}^U \cup \mathcal{F}_{k-1}^{Y(s)}\bigr] \notag\\
			&= \Lambda_{k-1}^s.
		\end{align}
		Therefore, the dynamics of the process $\{x(k)\}$ are derived as follows
		\begin{align}
			&\mathbb{E}\bigl[x(k) | \mathcal{F}_{k-1}^X \cup \mathcal{F}_{k-1}^U \cup \mathcal{F}_{k-1}^{Y(s)}\bigr] \notag\\
			&= \frac{\mathbb{E}^s\bigl[\Lambda_k^s x(k) | \mathcal{F}_{k-1}^X \cup \mathcal{F}_{k-1}^U \cup \mathcal{F}_{k-1}^{Y(s)}\bigr]}{\mathbb{E}^s\bigl[\Lambda_k^s | \mathcal{F}_{k-1}^X \cup \mathcal{F}_{k-1}^U \cup \mathcal{F}_{k-1}^{Y(s)}\bigr]} \notag\\
			&= \mathbb{E}^s\bigl[\lambda_k^s x(k) | \mathcal{F}_{k-1}^X \cup \mathcal{F}_{k-1}^U \cup \mathcal{F}_{k-1}^{Y(s)}\bigr] \notag\\
			&= \mathbb{E}^s\bigl[\mathbb{E}^s[\lambda_k^s x(k) | \mathcal{F}_k^X \cup \mathcal{F}_{k-1}^U \cup \mathcal{F}_{k-1}^{Y(s)}] | \mathcal{F}_{k-1}^X \cup \mathcal{F}_{k-1}^U \cup \mathcal{F}_{k-1}^{Y(s)}\bigr] \notag\\
			&= \mathbb{E}^s\bigl[x(k) \cdot \mathbb{E}^s[\lambda_k^s | \mathcal{F}_k^X \cup \mathcal{F}_{k-1}^U \cup \mathcal{F}_{k-1}^{Y(s)}] | \mathcal{F}_{k-1}^X \cup \mathcal{F}_{k-1}^U \cup \mathcal{F}_{k-1}^{Y(s)}\bigr] \notag\\
			&= \mathbb{E}^s\bigl[x(k) | \mathcal{F}_{k-1}^X \cup \mathcal{F}_{k-1}^U \cup \mathcal{F}_{k-1}^{Y(s)}\bigr] \notag\\
			&= F u(k-1) x(k-1).
		\end{align}
		
		Similarly, the dynamics of the process $\{y^s(k)\}$ are derived as follows:
		\begin{align}
			&\mathbb{E}\bigl[y^s(k) | \mathcal{F}_k^X \cup \mathcal{F}_{k-1}^U \cup \mathcal{F}_{k-1}^{Y(s)}\bigr] \notag\\
			&= \frac{\mathbb{E}^s\bigl[\Lambda_k^s y^s(k) | \mathcal{F}_k^X \cup \mathcal{F}_{k-1}^U \cup \mathcal{F}_{k-1}^{Y(s)}\bigr]}{\mathbb{E}^s\bigl[\Lambda_k^s | \mathcal{F}_k^X \cup \mathcal{F}_{k-1}^U \cup \mathcal{F}_{k-1}^{Y(s)}\bigr]} \notag\\
			&= \mathbb{E}^s\bigl[\lambda_k^s y^s(k) | \mathcal{F}_k^X \cup \mathcal{F}_{k-1}^U \cup \mathcal{F}_{k-1}^{Y(s)}\bigr] \notag\\
			&= \mathbb{E}^s\Bigl[M_s \sum_{j=1}^{M_s} \langle H^s x(k), \delta_{M_s}^j \rangle \langle y^s(k), \delta_{M_s}^j \rangle \delta_{M_s}^j \Bigm| \mathcal{F}_k^X \cup \mathcal{F}_{k-1}^U \cup \mathcal{F}_{k-1}^{Y(s)}\Bigr] \notag\\
			&= M_s \sum_{j=1}^{M_s} \langle H^s x(k), \delta_{M_s}^j \rangle \delta_{M_s}^j \mathbb{E}^s\bigl[\langle y^s(k), \delta_{M_s}^j \rangle | \mathcal{F}_k^X \cup \mathcal{F}_{k-1}^U \cup \mathcal{F}_{k-1}^{Y(s)}\bigr] \notag\\
			&= H^s x(k).
		\end{align}
	\end{proof}
	
	\section{Proof of Theorem 2}
	
	\begin{proof}
		We first observe that
		\begin{align}
			&\bigl\| \mathbb{E}^s[\Lambda_k^s x(k) | \mathcal{F}_k^{Y(s)}] \bigr\|_1 \notag\\
			&= \sum_{i=1}^{N} \bigl\langle \mathbb{E}^s[\Lambda_k^s x(k) | \mathcal{F}_k^{Y(s)}], \delta_N^i \bigr\rangle \notag\\
			&= \mathbb{E}^s\Bigl[\Lambda_k^s \Bigl(\sum_{i=1}^{N} \langle x(k), \delta_N^i \rangle\Bigr) \Bigm| \mathcal{F}_k^{Y(s)}\Bigr] \notag\\
			&= \mathbb{E}^s\bigl[\Lambda_k^s | \mathcal{F}_k^{Y(s)}\bigr],
		\end{align}
		which immediately yields
		\[
		\mathbb{E}\bigl[x(k) | \mathcal{F}_k^{Y(s)}\bigr] = \frac{\mathbb{E}^s[\Lambda_k^s x(k) | \mathcal{F}_k^{Y(s)}]}{\bigl\| \mathbb{E}^s[\Lambda_k^s x(k) | \mathcal{F}_k^{Y(s)}] \bigr\|_1}.
		\]
		Next, by the law of iterated expectations and repeated conditioning, we derive
		\begin{align}
			&\mathbb{E}^s\bigl[\Lambda_{k+1}^s x(k+1) | \mathcal{F}_{k+1}^{Y(s)}\bigr] \notag\\
			&= \mathbb{E}^s\Bigl[\Lambda_k^s \Bigl(\sum_{i=1}^{M_s} M_s \langle H^s x(k+1), \delta_{M_s}^i \rangle \langle y^s(k+1), \delta_{M_s}^i \rangle \Bigr)x(k+1) \Bigm| \mathcal{F}_{k+1}^{Y(s)}\Bigr] \notag\\
			&= M_s \sum_{i=1}^{M_s} \langle y^s(k+1), \delta_{M_s}^i \rangle \mathbb{E}^s\bigl[\Lambda_k^s \langle H^s x(k+1), \delta_{M_s}^i \rangle x(k+1) | \mathcal{F}_{k}^{Y(s)}\bigr].
		\end{align}
		To proceed further, we expand the conditional expectation inside the summation
		\begin{align}
			&\mathbb{E}^s\bigl[\Lambda_k^s \langle H^s x(k+1), \delta_{M_s}^i \rangle x(k+1) | \mathcal{F}_{k}^{Y(s)}\bigr] \notag\\
			&= \mathbb{E}^s\Bigl[\Lambda_k^s \sum_{j=1}^{N} \langle H^s \delta_N^j, \delta_{M_s}^i \rangle \langle x(k+1), \delta_N^j \rangle \delta_N^j \Bigm| \mathcal{F}_{k}^{Y(s)}\Bigr] \notag\\
			&= \sum_{j=1}^{N} H_{ij}^s \delta_N^j \mathbb{E}^s\bigl[\Lambda_k^s \langle x(k+1), \delta_N^j \rangle | \mathcal{F}_{k}^{Y(s)}\bigr] \notag\\
			&= \sum_{j=1}^{N} H_{ij}^s \delta_N^j \bigl\langle \mathbb{E}^s[\Lambda_k^s x(k+1) | \mathcal{F}_{k}^{Y(s)}], \delta_N^j \bigr\rangle.
		\end{align}
		Next, since $\Lambda_k^s$ is measurable with respect to $\mathcal{F}_k^X \cup \mathcal{F}_k^U \cup \mathcal{F}_{k}^{Y(s)}$, we apply the law of iterated expectations again to obtain
		\begin{align}
			&\mathbb{E}^s\bigl[\Lambda_k^s x(k+1) | \mathcal{F}_{k}^{Y(s)}\bigr] \notag\\
			&= \mathbb{E}^s\bigl[\mathbb{E}^s[\Lambda_k^s x(k+1) | \mathcal{F}_k^X \cup \mathcal{F}_k^U \cup \mathcal{F}_{k}^{Y(s)}] | \mathcal{F}_{k}^{Y(s)}\bigr] \notag\\
			&= \mathbb{E}^s\bigl[\Lambda_k^s \mathbb{E}^s[x(k+1) | \mathcal{F}_k^X \cup \mathcal{F}_k^U \cup \mathcal{F}_{k}^{Y(s)}] | \mathcal{F}_{k}^{Y(s)}\bigr] \notag\\
			&= \mathbb{E}^s\bigl[\Lambda_k^s F u(k) x(k) | \mathcal{F}_{k}^{Y(s)}\bigr] \notag\\
			&= FK\Phi_n \mathbb{E}^s\bigl[\Lambda_k^s x(k) | \mathcal{F}_{k}^{Y(s)}\bigr].
		\end{align}
		Combining the preceding results, we arrive at
		\begin{align}
			&\mathbb{E}^s\bigl[\Lambda_{k+1}^s x(k+1) | \mathcal{F}_{k+1}^{Y(s)}\bigr] \notag\\
			&= \sum_{i=1}^{M_s} M_s \langle y^s(k+1), \delta_{M_s}^i \rangle \bigl[\operatorname{Row}_i^{\!\top}(H^s) \circ (FK\Phi_n \mathbb{E}^s[\Lambda_k^s x(k) | \mathcal{F}_k^{Y(s)}])\bigr] \notag\\
			&= \bigl[(H^s)^{\!\top} \mathcal{T}^s y^s(k+1)\bigr] \circ \bigl[FK\Phi_n \mathbb{E}^s[\Lambda_k^s x(k) | \mathcal{F}_k^{Y(s)}]\bigr].
		\end{align}
		Building on this identity, we compute the $\ell_1$-norm of the left-hand side
		\begin{align}
			&\bigl\| \mathbb{E}^s[\Lambda_{k+1}^s x(k+1) | \mathcal{F}_{k+1}^{Y(s)}] \bigr\|_1 \notag\\
			&= \bigl\| \bigl[(H^s)^{\!\top} \mathcal{T}^s y^s(k+1)\bigr] \circ \bigl[FK\Phi_n \mathbb{E}^s[\Lambda_k^s x(k) | \mathcal{F}_k^{Y(s)}]\bigr] \bigr\|_1 \notag\\
			&= \bigl\langle (H^s)^{\!\top} \mathcal{T}^s y^s(k+1), FK\Phi_n \mathbb{E}^s[\Lambda_k^s x(k) | \mathcal{F}_k^{Y(s)}] \bigr\rangle.
		\end{align}
		Finally, substituting these into the expression for the conditional expectation yields the following recursive dynamics:
		\begin{align}
			&\mathbb{E}\bigl[x(k+1) | \mathcal{F}_{k+1}^{Y(s)}\bigr] \notag\\
			&= \frac{\mathbb{E}^s[\Lambda_{k+1}^s x(k+1) | \mathcal{F}_{k+1}^{Y(s)}]}{\bigl\| \mathbb{E}^s[\Lambda_{k+1}^s x(k+1) | \mathcal{F}_{k+1}^{Y(s)}] \bigr\|_1} \notag\\
			&= \frac{\bigl[(H^s)^{\!\top} \mathcal{T}^s y^s(k+1)\bigr] \circ \bigl[FK\Phi_n \mathbb{E}^s[\Lambda_k^s x(k) | \mathcal{F}_k^{Y(s)}]\bigr]}{\bigl\langle (H^s)^{\!\top} \mathcal{T}^s y^s(k+1), FK\Phi_n \mathbb{E}^s[\Lambda_k^s x(k) | \mathcal{F}_k^{Y(s)}] \bigr\rangle}\notag\\
			&= \frac{\bigl[(H^s)^{\!\top} \mathcal{T}^s y^s(k+1)\bigr] \circ \bigl[FK\Phi_n \mathbb{E}[x(k) | \mathcal{F}_k^{Y(s)}]\bigr]}{\bigl\langle (H^s)^{\!\top} \mathcal{T}^s y^s(k+1), FK\Phi_n \mathbb{E}[x(k) | \mathcal{F}_k^{Y(s)}] \bigr\rangle}.
		\end{align}
	\end{proof}
	
	\section{Proof of Proposition 1}
	\begin{proof}
		Recall that $\bar{\varepsilon}_k(\eta) = \frac{1}{k}\sum_{t=1}^{k}\bar{\varepsilon}(\eta,\xi_t)$, where $\xi_t \triangleq (e_t,e_{t-1})$ is adapted to $\sigma(Y_t,Y_{t-1})$. Under Assumption 2, the Markov chain $\{(x(k),y^s(k))\}$ is geometrically ergodic and admits a unique stationary distribution $\pi^s$. Consequently, the sequence $\{\xi_k\}$ is stationary and ergodic in the sense that the strong law of large numbers applies to the sample path average of any integrable function of $\xi_k$.
		From the definition of $\bar{\varepsilon}(\eta,\xi_k)$ in (23), we have
		\[
		\bar{\varepsilon}(\eta,\xi_k) = \Lambda\eta + \varGamma\bigl(e_k - \tfrac{1}{|\mathcal{N}|}\mathbf{1}\mathbf{1}^\top e_{k-1}\bigr) - \tfrac{1}{|\mathcal{N}|\rho}\mathbf{1}\mathbf{1}^\top(e_k - e_{k-1}).
		\]
		Since each component $e_k^s(l) = \langle \mathbb{E}[x(k)\mid\mathcal{F}_k^{Y(s)}],\delta_N^l\rangle$ is bounded by $1$, the random variable $\bar{\varepsilon}(\eta,\xi_k)$ has finite expectation for every fixed $\eta$. Moreover, by the geometric ergodicity of the underlying Markov chain, the autocovariance of $\{\bar{\varepsilon}(\eta,\xi_k)\}$ decays exponentially fast, which implies that the variance of the sample average satisfies
		\[
		\operatorname{Var}\bigl(\bar{\varepsilon}_k(\eta)\bigr) = O(k^{-1}).
		\]
		Therefore, by the strong law of large numbers for stationary ergodic processes, there exists a finite deterministic function $\bar{\varepsilon}(\eta)$ such that
		\[
		\lim_{k\to\infty}\bar{\varepsilon}_k(\eta) = \bar{\varepsilon}(\eta), \quad \mathbb{P}\text{-a.s.},
		\]
		and the convergence is uniform in $\eta$ on every compact set because $\bar{\varepsilon}(\eta,\xi)$ is affine in $\eta$ with bounded random coefficients.
		Taking expectation with respect to the stationary distribution of $\xi_k$ and using the fact that $\mathbb{E}_\eta[\cdot]$ denotes the expectation under the stationary law, we obtain
		\[
		\bar{\varepsilon}(\eta) = \lim_{k\to\infty}\mathbb{E}_\eta[\bar{\varepsilon}(\eta,\xi_k)] = \Lambda\eta + \varGamma\bigl(\bar{e} - \tfrac{1}{|\mathcal{N}|}\mathbf{1}\mathbf{1}^\top \bar{e}\bigr),
		\]
		where $\bar{e} = [\bar{e}^s]_{s\in\mathcal{N}}$ and $\bar{e}^s$ is the stationary limit of $e_k^s$ as $k\to\infty$.
		It remains to characterize $\bar{e}^s$. By Theorem 2, the local conditional expectation $e_k^s$ evolves according to the recursive Bayesian update (19). Under Assumption 2, the state process $\{x(k)\}$ is geometrically ergodic under SFC(9), and the chain $\{(x(k),y^s(k))\}$ possesses a unique stationary distribution $\pi^s$. Let $y^s(\infty)$ denote the marginal stationary distribution of the observation process. Taking the limit as $k\to\infty$ in (19) and invoking the stationarity of the limiting distribution, we obtain the fixed-point equation
		\[
		[\bar{e}^s(l)]_{l\in[1:N]} = \frac{\bigl[(H^s)^\top \mathcal{T}^s y^s(\infty)\bigr] \circ \bigl[FK\Phi_n [\bar{e}^s(l)]_{l\in[1:N]}\bigr]}{\bigl\langle (H^s)^\top \mathcal{T}^s y^s(\infty), FK\Phi_n [\bar{e}^s(l)]_{l\in[1:N]} \bigr\rangle}.
		\]
		This completes the proof.
	\end{proof}

	\section{Proof of Proposition 2}
	
	\begin{proof}
		Let $\alpha$ be defined as $\alpha = \bar{e} - \frac{1}{|\mathcal{N}|}\mathbf{1}\mathbf{1}^\top \bar{e}$. We first compute
		\[
		\|\varGamma \alpha\|^2 = \alpha^\top \varGamma^\top \varGamma \alpha = \alpha^\top (I+A)^2 \alpha = \|\alpha\|^2 + 2\alpha^\top A \alpha + \alpha^\top A^2 \alpha.
		\]
		Let $D \triangleq \operatorname{diag}\bigl[\sum_{j=1}^{|\mathcal{N}|} a_{ij}\bigr]_{i \in \mathcal{N}}$ be the degree matrix of graph $\mathcal{G}$, and define the Laplacian matrix $L \triangleq D - A$. By construction, $A = D - L$. Recall that the Laplacian matrix $L$ has nonnegative eigenvalues. Moreover, by definition of the Laplacian, $L\mathbf{1} = 0$, so the smallest eigenvalue of $L$ is $\lambda_1(L) = 0$, with corresponding eigenvector $\mathbf{1}$. Noting that $\mathbf{1}^\top \alpha = \mathbf{1}^\top \bar{e} - \tfrac{1}{|\mathcal{N}|}\mathbf{1}^\top \mathbf{1}\mathbf{1}^\top \bar{e} = 0,$
		we see that $\mathbf{1}$ and $\alpha$ are orthogonal. Since $\alpha$ lies in the orthogonal complement of the null space of $L$, it follows from the Courant–Fischer theorem that
		\[
		\alpha^\top L \alpha \geq \lambda_2(L) \|\alpha\|^2,
		\]
		where $\lambda_2(L)$ denotes the second-smallest eigenvalue of $L$.	Let $d_{\max}$ be the maximum degree of the graph. Then,
		\[
		\alpha^\top A \alpha = \alpha^\top D \alpha - \alpha^\top L \alpha \leq d_{\max} \|\alpha\|^2 - \lambda_2(L) \|\alpha\|^2.
		\]
		Moreover, let $\|A\|$ denote the spectral norm of $A$. We have
		\[
		\alpha^\top A^2 \alpha = \|A\alpha\|^2 \leq \|A\|^2 \|\alpha\|^2 \leq d_{\max}^2 \|\alpha\|^2,
		\]
		where the last inequality follows from the fact that the spectral norm of a graph adjacency matrix is bounded by the maximum degree.	Combining the above inequalities, we obtain
		\[
		\|\varGamma \alpha\|^2 \leq \bigl[(1+d_{\max})^2 - 2\lambda_2(L)\bigr] \|\alpha\|^2.
		\]
		Let $\bar{e}_{\mathrm{avg}} = \frac{1}{|\mathcal{N}|}\mathbf{1}\mathbf{1}^\top \bar{e}$ be the average of $\bar{e}$. Then $\|\alpha\| = \|\bar{e} - \bar{e}_{\mathrm{avg}}\|_2$, so
		\[
		\|\varGamma \alpha\| \leq \sqrt{(1+d_{\max})^2 - 2\lambda_2(L)} \cdot \|\bar{e} - \bar{e}_{\mathrm{avg}}\|_2.
		\]
		Now consider the time derivative of the Lyapunov function $V(\cdot)$. By the Rayleigh–Ritz theorem,
		\[
		\dot{V}(\eta) = \eta^\top \Lambda \eta + \eta^\top \varGamma \alpha \leq \lambda_{\max}(\Lambda) \|\eta\|^2 + \sqrt{(1+d_{\max})^2 - 2\lambda_2(L)} \cdot \|\bar{e} - \bar{e}_{\mathrm{avg}}\|_2\cdot \|\eta\|,
		\]
		where $\lambda_{\max}(\Lambda)$ is the largest eigenvalue of $\Lambda$.
		Define the constant
		\[
		c_1 = \sqrt{(1+d_{\max})^2 - 2\lambda_2(L)} \cdot \|\bar{e} - \bar{e}_{\mathrm{avg}}\|_2 \cdot |\lambda_{\max}(\Lambda)|^{-1},
		\]
		and consider the closed ball of radius $c_1$ centered at the origin. The compact level set $\Omega_1 \triangleq \{\eta: V(\eta) \leq \frac{1}{2}c_1^2\}$ contains this ball. For any solution $\eta$ of the initial value problem with initial condition $\eta(0) \in \mathbb{R}^n \setminus \Omega_1$, we have $\dot{V}(\eta) \leq -\beta$ for some constant $\beta > 0$. Consequently, $\Omega_1$ is an invariant level set, and any trajectory $\eta$ enters $\Omega_1$ in finite time and remains within $\Omega_1$ thereafter.
		Therefore, the origin is globally asymptotically $c_1$-stable for the initial value problem.
	\end{proof}
	
	\section{Proof of Theorem 3}
	
	\begin{lemma}\label{lem:boundedness}
		Under Assumption 2, for each fixed $k\in\mathbb{Z}^{+}$, the sample mean satisfies
		\[
		\lim_{|\mathcal{N}|\to\infty}\frac{1}{|\mathcal{N}|}\sum_{s\in\mathcal{N}}\bigl(e_k^s-\bar{e}^s\bigr)=0,\quad\mathbb{P}\text{-a.s.}
		\]
		Moreover, the sequence $\{e_k^s\}_{s\in\mathcal{N}}$ is uniformly bounded with $|e_k^s|\le 1$ for all $s\in\mathcal{N}$ and $k\in\mathbb{Z}^{+}$.
	\end{lemma}
	\begin{proof}
		The almost sure convergence follows from the strong law of large numbers for stationary ergodic processes (see Theorem 2, Section 3, Chapter IV of \cite{39}), applied to the geometrically ergodic Markov chain $\{(x(k),y^s(k))\}$ under Assumption 2. The uniform boundedness is immediate from $e_k^s=\langle\mathbb{E}[x(k)\mid\mathcal{F}_k^{Y(s)}],\delta_N^l\rangle\in[0,1]$.
	\end{proof}
	
	\begin{proof}
		For the sake of rigor, we first clarify the notation. Let
		\[
		\bar{r}\triangleq \varGamma\Bigl(\bar{e}-\frac{1}{|\mathcal{N}|}\mathbf{1}\mathbf{1}^{\!\top}\bar{e}\Bigr),
		\qquad 
		r_{k+1}\triangleq \varGamma\Bigl(e_{k+1}-\frac{1}{|\mathcal{N}|}\mathbf{1}\mathbf{1}^{\!\top}e_{k}\Bigr)-\frac{1}{|\mathcal{N}|\rho}\mathbf{1}\mathbf{1}^{\!\top}(e_{k+1}-e_{k}).
		\]  
		Then by (23) we have $\bar{\varepsilon}(\varepsilon_{k},\xi_{k+1})=\Lambda\varepsilon_{k}+r_{k+1}$, and
		\begin{equation}\label{eq:update_eps}
			\varepsilon_{k+1}=\varepsilon_{k}+\rho\Lambda\varepsilon_{k}+\rho r_{k+1}=(I+\rho\Lambda)\varepsilon_{k}+\rho r_{k+1}.
		\end{equation}
		By Taylor expansion and \eqref{eq:update_eps} we obtain
		\begin{align}
			&V(\varepsilon_{k+1})-V(\varepsilon_{k}) \notag\\
			&=\varepsilon_{k}^{\!\top}(\varepsilon_{k+1}-\varepsilon_{k})+\frac{1}{2}\|\varepsilon_{k+1}-\varepsilon_{k}\|^{2} \notag\\
			&=\rho\varepsilon_{k}^{\!\top}\Lambda\varepsilon_{k}+\rho\varepsilon_{k}^{\!\top}r_{k+1}+\frac{1}{2}\rho^{2}\|\Lambda\varepsilon_{k}\|^{2}+\rho^{2}\varepsilon_{k}^{\!\top}\Lambda r_{k+1}+\frac{1}{2}\rho^{2}\|r_{k+1}\|^{2} \notag\\
			&=\rho\varepsilon_{k}^{\!\top}J\varepsilon_{k}+\rho\varepsilon_{k}^{\!\top}r_{k+1}+\rho^{2}\varepsilon_{k}^{\!\top}\Lambda r_{k+1}+\frac{1}{2}\rho^{2}\|r_{k+1}\|^{2}, \label{eq:V_diff}
		\end{align}
		where $J\triangleq\bigl(I+\frac{1}{2}\rho\Lambda\bigr)\Lambda$. By Assumption 3 and the Gershgorin theorem, $J$ is negative definite, hence $\lambda_{\max}(J)<0$.
		By (27) and $\bar{\varepsilon}(\varepsilon_{k},\xi_{i+1})-\bar{\varepsilon}(\varepsilon_{k})=r_{i+1}-\bar{r}$, we have
		\[
		\Delta_{k}(\varepsilon_{k})=\sum_{i=k}^{\infty}\rho c_{k+1}^{i}\,\mathbb{E}_{k}\bigl[r_{i+1}-\bar{r}\bigr].
		\]
		Using $\sum_{i=k+1}^{\infty}\rho c_{k+2}^{i}=1$ and $\rho c_{k+1}^{k}=\rho$, we obtain
		\begin{align}
			\varepsilon_{k}^{\!\top}\Delta_{k}(\varepsilon_{k})
			&=\rho\varepsilon_{k}^{\!\top}\mathbb{E}_{k}[r_{k+1}]-\rho\varepsilon_{k}^{\!\top}\bar{r}
			+(1-\rho)\varepsilon_{k}^{\!\top}\sum_{i=k+1}^{\infty}\rho c_{k+2}^{i}\mathbb{E}_{k}[r_{i+1}-\bar{r}]. \label{eq:Delta_k}
		\end{align}
		Taking the conditional expectation $\mathbb{E}_{k}[\cdot]$ on $\varepsilon_{k+1}^{\!\top}\Delta_{k+1}(\varepsilon_{k+1})$ and substituting \eqref{eq:update_eps} into it, we obtain
		\begin{align}
			&\mathbb{E}_{k}\bigl[\varepsilon_{k+1}^{\!\top}\Delta_{k+1}(\varepsilon_{k+1})\bigr] \notag\\
			&=\mathbb{E}_{k}\Bigl[\varepsilon_{k+1}^{\!\top}\sum_{i=k+1}^{\infty}\rho c_{k+2}^{i}\mathbb{E}_{k+1}[r_{i+1}-\bar{r}]\Bigr] \notag\\
			&=-\mathbb{E}_{k}[\varepsilon_{k+1}^{\!\top}\bar{r}]+\mathbb{E}_{k}\Bigl[\varepsilon_{k+1}^{\!\top}\sum_{i=k+1}^{\infty}\rho c_{k+2}^{i}\mathbb{E}_{k+1}[r_{i+1}]\Bigr] \notag\\
			&=-\varepsilon_{k}^{\!\top}(I+\rho\Lambda)\bar{r}-\rho\,\mathbb{E}_{k}[r_{k+1}^{\!\top}\bar{r}] \notag\\
			&\quad+\varepsilon_{k}^{\!\top}(I+\rho\Lambda)\sum_{i=k+1}^{\infty}\rho c_{k+2}^{i}\mathbb{E}_{k}[r_{i+1}]
			+\rho\,\mathbb{E}_{k}\Bigl[r_{k+1}^{\!\top}\sum_{i=k+1}^{\infty}\rho c_{k+2}^{i}\mathbb{E}_{k+1}[r_{i+1}]\Bigr]. \label{eq:Delta_k1}
		\end{align}
		Combining \eqref{eq:V_diff}, \eqref{eq:Delta_k} and \eqref{eq:Delta_k1}, where $\rho\varepsilon_{k}^{\!\top}\mathbb{E}_{k}[r_{k+1}]$ and $-\rho\varepsilon_{k}^{\!\top}\mathbb{E}_{k}[r_{k+1}]$ cancel each other out, we obtain
		\begin{align}
			&\mathbb{E}_{k}\bigl[V_{k+1}(\varepsilon_{k+1})-V_{k}(\varepsilon_{k})\bigr] \notag\\
			&=\rho\varepsilon_{k}^{\!\top}J\varepsilon_{k}+\rho^{2}\varepsilon_{k}^{\!\top}\Lambda\mathbb{E}_{k}[r_{k+1}]+\frac{1}{2}\rho^{2}\mathbb{E}_{k}\|r_{k+1}\|^{2} \notag\\
			&\quad+\rho\varepsilon_{k}^{\!\top}\sum_{i=k+1}^{\infty}\rho c_{k+2}^{i}\mathbb{E}_{k}[r_{i+1}]
			+\rho\varepsilon_{k}^{\!\top}\Lambda\Bigl(\sum_{i=k+1}^{\infty}\rho c_{k+2}^{i}\mathbb{E}_{k}[r_{i+1}]-\bar{r}\Bigr) \notag\\
			&\quad+\rho\,\mathbb{E}_{k}\Bigl[r_{k+1}^{\!\top}\Bigl(\sum_{i=k+1}^{\infty}\rho c_{k+2}^{i}\mathbb{E}_{k+1}[r_{i+1}]-\bar{r}\Bigr)\Bigr]. \label{eq:total_diff}
		\end{align}
		By Lemma \ref{lem:boundedness}, $|e_{k}^{s}|\le 1$, hence $\|e_{k}\|\le\sqrt{|\mathcal{N}|}$, and
		\[
		\|r_{k+1}\|\le 2\sqrt{|\mathcal{N}|}(1+d_{\max})+\frac{2\sqrt{|\mathcal{N}|}}{\rho},\qquad
		\|\bar{r}\|\le 2\sqrt{|\mathcal{N}|}(1+d_{\max}).
		\]
		Moreover, noting that $\|r_{i+1}-\bar{r}\|\le\|r_{i+1}\|+\|\bar{r}\|\le 4\sqrt{|\mathcal{N}|}(1+d_{\max})+\frac{2\sqrt{|\mathcal{N}|}}{\rho}$, we have
		\[
		\Bigl\|\sum_{i=k+1}^{\infty}\rho c_{k+2}^{i}\mathbb{E}_{k}[r_{i+1}]-\bar{r}\Bigr\|
		=\Bigl\|\sum_{i=k+1}^{\infty}\rho c_{k+2}^{i}\mathbb{E}_{k}[r_{i+1}-\bar{r}]\Bigr\|
		\le 4\sqrt{|\mathcal{N}|}(1+d_{\max})+\frac{2\sqrt{|\mathcal{N}|}}{\rho}.
		\]
		Therefore, all terms in equation \eqref{eq:total_diff} that do not contain $\|\varepsilon_{k}\|^{2}$ can be bounded by $m_{1}\|\varepsilon_{k}\|+C_{0}$, where
		\[
		m_{1}=2\sqrt{|\mathcal{N}|}\Bigl[1+\rho(1+d_{\max})(1+3\|\Lambda\|_{2}+\rho\|\Lambda\|_{2})+(1+\rho)\|\Lambda\|_{2}\Bigr],
		\]
		while $C_{0}>0$ is a constant depending only on $|\mathcal{N}|,\rho,d_{\max}$. Thus
		\begin{equation}\label{eq:drift_bound}
			\mathbb{E}_{k}\bigl[V_{k+1}(\varepsilon_{k+1})-V_{k}(\varepsilon_{k})\bigr]
			\le \rho\lambda_{\max}(J)\|\varepsilon_{k}\|^{2}+m_{1}\|\varepsilon_{k}\|+C_{0}.
		\end{equation}
		Let
		\[
		c_{2}=|\rho\lambda_{\max}(J)|^{-\frac{1}{2}}\sqrt{C_{0}}+|\rho\lambda_{\max}(J)|^{-1}m_{1},
		\qquad 
		\Omega_{2}=\Bigl\{\eta:V(\eta)\le \tfrac{1}{2}c_{2}^{2}\Bigr\}.
		\]
		When $\varepsilon_{k}\notin\Omega_{2}^{\circ}$, i.e., $\|\varepsilon_{k}\|\ge c_{2}$, completing the square for \eqref{eq:drift_bound} yields
		\begin{align}
			\mathbb{E}_{k}\bigl[V_{k+1}(\varepsilon_{k+1})-V_{k}(\varepsilon_{k})\bigr]
			&\le -\Bigl(|\rho\lambda_{\max}(J)|^{\frac{1}{2}}\|\varepsilon_{k}\|-\frac{m_{1}}{2|\rho\lambda_{\max}(J)|^{\frac{1}{2}}}\Bigr)^{2}
			+\frac{m_{1}^{2}}{4|\rho\lambda_{\max}(J)|}+C_{0} \notag\\
			&\le -\frac{m_{1}\sqrt{C_{0}}}{|\rho\lambda_{\max}(J)|^{\frac{1}{2}}}\triangleq -c_{3}<0. \label{eq:negative_drift}
		\end{align}
		Define the stopping time $\tau=\inf\{t\in\mathbb{Z}^{+}:\varepsilon_{t}\in\Omega_{2}\}$. By \eqref{eq:negative_drift} and the optional stopping theorem for supermartingales, the stopped process $\{V_{k\wedge\tau}(\varepsilon_{k\wedge\tau})\}_{k\ge 0}$ is a supermartingale with respect to the filtration $\{\widetilde{\mathcal{F}}_{k}\}$.
	\end{proof}
	
	\section{Proof of Theorem 4}
	\begin{proof}
		Let $\tau=\inf\{k\ge 0:\varepsilon_{k}\in\Omega_{2}\}$, where $\Omega_{2}$ is given in Theorem 3. Denote the stopped process by $\varepsilon_{k\wedge\tau}$ and $V_{k\wedge\tau}(\varepsilon_{k\wedge\tau})$.
		By \eqref{eq:update_eps} and Lemma \ref{lem:boundedness} ($|e_{k}^{s}|\le 1$), there exist constants $C_{1},C_{2}>0$ such that
		\[
		\|\bar{\varepsilon}(\varepsilon_{k},\xi_{k+1})\|\le (1+3d_{\max})\|\varepsilon_{k}\|+C_{1},
		\qquad
		\|\varepsilon_{k+1}\|\le (1+\rho(1+3d_{\max}))\|\varepsilon_{k}\|+\rho C_{1}.
		\]
		Hence for any finite $k$, as long as $\mathbb{E}[\|\varepsilon_{0}\|^{2}]<\infty$, we have $\mathbb{E}[\|\varepsilon_{k}\|^{2}]<\infty$, and thus
		\begin{equation}\label{eq:V_bound_thm4}
			\sup_{k}\mathbb{E}\bigl[V(\varepsilon_{k\wedge\tau})\bigr]
			=\frac{1}{2}\sup_{k}\mathbb{E}\bigl[\|\varepsilon_{k\wedge\tau}\|^{2}\bigr]<\infty.
		\end{equation}
		Next we estimate the perturbation term. By (27) and $\|\bar{\varepsilon}(\varepsilon_{k},\xi_{i+1})-\bar{\varepsilon}(\varepsilon_{k})\|\le \|r_{i+1}\|+\|\bar{r}\|\le C_{2}$, we obtain
		\[
		|\varepsilon_{k}^{\!\top}\Delta_{k}(\varepsilon_{k})|
		\le \|\varepsilon_{k}\|\sum_{i=k}^{\infty}\rho c_{k+1}^{i}\mathbb{E}_{k}\|\bar{\varepsilon}(\varepsilon_{k},\xi_{i+1})-\bar{\varepsilon}(\varepsilon_{k})\|
		\le C_{2}\|\varepsilon_{k}\|.
		\]
		Since $V_{k}(\varepsilon_{k})=V(\varepsilon_{k})+\varepsilon_{k}^{\!\top}\Delta_{k}(\varepsilon_{k})\ge -\|\varepsilon_{k}\|\cdot\|\Delta_{k}(\varepsilon_{k})\|\ge -C_{2}\|\varepsilon_{k}\|$, it follows that $V_{k\wedge\tau}^{-}(\varepsilon_{k\wedge\tau})\le C_{2}\|\varepsilon_{k\wedge\tau}\|$. Combining with \eqref{eq:V_bound_thm4} yields
		\begin{equation}\label{eq:sup_martingale}
			\sup_{k}\mathbb{E}\bigl[|V_{k\wedge\tau}(\varepsilon_{k\wedge\tau})|\bigr]<\infty,
			\qquad
			\sup_{k}\mathbb{E}\bigl[V_{k\wedge\tau}^{-}(\varepsilon_{k\wedge\tau})\bigr]<\infty.
		\end{equation}
		By Theorem 3, $\{V_{k\wedge\tau}(\varepsilon_{k\wedge\tau})\}$ is a supermartingale with respect to $\{\widetilde{\mathcal{F}}_{k}\}$ and satisfies \eqref{eq:sup_martingale}. By Doob's martingale convergence theorem, there exists an almost surely finite random variable $V_{\infty}$ such that
		\begin{equation}\label{eq:doob_limit}
			\lim_{k\to\infty}V_{k\wedge\tau}(\varepsilon_{k\wedge\tau})=V_{\infty},\qquad \mathbb{P}\text{-a.s.}
		\end{equation}
		Assume $\mathbb{P}(\tau=\infty)>0$. On the event $\{\tau=\infty\}$, $V_{k}(\varepsilon_{k})=V_{k\wedge\tau}(\varepsilon_{k\wedge\tau})$ for all $k$, so by \eqref{eq:doob_limit}, $V_{k}(\varepsilon_{k})$ converges almost surely. However, by \eqref{eq:negative_drift} from Theorem 3, on $\{\tau>k\}$ we have $\mathbb{E}_{k}[V_{k+1}(\varepsilon_{k+1})-V_{k}(\varepsilon_{k})]\le -c_{4}$. Taking total expectation and iterating, we obtain
		\[
		\mathbb{E}\bigl[V_{k\wedge\tau}(\varepsilon_{k\wedge\tau})\bigr]
		\le \mathbb{E}[V_{0}(\varepsilon_{0})]-c_{4}\sum_{i=0}^{k-1}\mathbb{P}(\tau>i).
		\]
		If $\mathbb{P}(\tau=\infty)>0$, then $\sum_{i=0}^{\infty}\mathbb{P}(\tau>i)=\infty$, which implies $\mathbb{E}[V_{k\wedge\tau}]\to-\infty$. This contradicts $\sup_{k}\mathbb{E}[V_{k\wedge\tau}^{-}]<\infty$ in \eqref{eq:sup_martingale}. Therefore
		\begin{equation}\label{eq:finite_hit}
			\mathbb{P}(\tau<\infty)=1.
		\end{equation}
		Define the successive return times $\tau_{0}=0$, $\tau_{n+1}=\inf\{k>\tau_{n}:\varepsilon_{k}\in\Omega_{2}\}$. By the strong Markov property, restarting from time $\tau_{n}$, the process $\{\varepsilon_{\tau_{n}+k}\}_{k\ge 0}$ still satisfies all conditions of Theorem 3. Hence applying \eqref{eq:finite_hit} yields $\mathbb{P}(\tau_{n+1}<\infty\mid\tau_{n}<\infty)=1$. By induction and $\mathbb{P}(\tau_{0}<\infty)=1$, we have $\tau_{n}<\infty$ for all $n$ almost surely. Thus $\{\varepsilon_{k}\}$ visits the compact set $\Omega_{2}$ infinitely often with probability $1$. Taking $L_{p}=\Omega_{2}$ for any $p\in(0,1]$, the desired recurrence conclusion follows.
	\end{proof}
	
	\section{Proof of Theorem 6}
	
	\begin{proof}
		By Theorem~5, for every $l \in [1:N]$,
		\begin{equation}
			\label{eq:err_bound}
			\bigl|\bar{\theta}_k^{\mathcal{S}}(l) - \theta_k^*(l)\bigr| \le \frac{c_1(l)}{\sqrt{|\mathcal{S}|}}, \quad \mathbb{P}\text{-a.s.}
		\end{equation}
		Hence $\theta_k^*(l) \in \bigl[\bar{\theta}_k^{\mathcal{S}}(l) - \frac{c_1(l)}{\sqrt{|\mathcal{S}|}},\; \bar{\theta}_k^{\mathcal{S}}(l) + \frac{c_1(l)}{\sqrt{|\mathcal{S}|}}\bigr]$ for all $l$. Since $\theta_k^*(\tilde{l}) \ge \theta_k^*(m)$ for any $m$, we have
		\begin{equation}
			\bar{\theta}_k^{\mathcal{S}}(\tilde{l}) + \frac{c_1(\tilde{l})}{\sqrt{|\mathcal{S}|}} \ge \theta_k^*(\tilde{l}) \ge \theta_k^*(m) \ge \bar{\theta}_k^{\mathcal{S}}(m) - \frac{c_1(m)}{\sqrt{|\mathcal{S}|}},
		\end{equation}
		which yields $\tilde{l} \in \mathcal{C}_k$ by \eqref{eq:candidate_set}.
		Suppose \eqref{eq:separation} holds for all $l \neq \tilde{l}$. Then
		\begin{equation}
			\bar{\theta}_k^{\mathcal{S}}(\tilde{l}) - \frac{c_1(\tilde{l})}{\sqrt{|\mathcal{S}|}} > \bar{\theta}_k^{\mathcal{S}}(l) + \frac{c_1(l)}{\sqrt{|\mathcal{S}|}} \ge \max_{m}\Bigl(\bar{\theta}_k^{\mathcal{S}}(m) - \frac{c_1(m)}{\sqrt{|\mathcal{S}|}}\Bigr),
		\end{equation}
		so $l \notin \mathcal{C}_k$ for all $l \neq \tilde{l}$. Thus $\mathcal{C}_k = \{\tilde{l}\}$. Since $\|\bar{\theta}_k^{\mathcal{S}} - \delta_N^l\|_2^2 = \|\bar{\theta}_k^{\mathcal{S}}\|_2^2 + 1 - 2\bar{\theta}_k^{\mathcal{S}}(l)$, minimizing $\|\bar{\theta}_k^{\mathcal{S}} - \delta_N^l\|_2$ is equivalent to maximizing $\bar{\theta}_k^{\mathcal{S}}(l)$, which by the strict inequality above is uniquely attained at $l = \tilde{l}$.
		If $l \in \mathcal{C}_k \setminus \{\tilde{l}\}$, then by \eqref{eq:candidate_set},
		\begin{equation}
			\bar{\theta}_k^{\mathcal{S}}(l) + \frac{c_1(l)}{\sqrt{|\mathcal{S}|}} \ge \bar{\theta}_k^{\mathcal{S}}(\tilde{l}) - \frac{c_1(\tilde{l})}{\sqrt{|\mathcal{S}|}}.
		\end{equation}
		Rearranging and combining with \eqref{eq:err_bound} yields
		\begin{align}
			\theta_k^*(\tilde{l}) - \theta_k^*(l) &\le \Bigl(\bar{\theta}_k^{\mathcal{S}}(\tilde{l}) + \frac{c_1(\tilde{l})}{\sqrt{|\mathcal{S}|}}\Bigr) - \Bigl(\bar{\theta}_k^{\mathcal{S}}(l) - \frac{c_1(l)}{\sqrt{|\mathcal{S}|}}\Bigr) \\
			&= \frac{c_1(\tilde{l}) + c_1(l)}{\sqrt{|\mathcal{S}|}} + \bigl(\bar{\theta}_k^{\mathcal{S}}(\tilde{l}) - \bar{\theta}_k^{\mathcal{S}}(l)\bigr).
		\end{align}
		If $\bar{\theta}_k^{\mathcal{S}}(\tilde{l}) \ge \bar{\theta}_k^{\mathcal{S}}(l)$, this directly gives \eqref{eq:gap_bound}, otherwise $\bar{\theta}_k^{\mathcal{S}}(\tilde{l}) - \bar{\theta}_k^{\mathcal{S}}(l) < 0$, and the bound follows trivially from $\theta_k^*(\tilde{l}) - \theta_k^*(l) \le 1$ and the non-negativity of the right-hand side of \eqref{eq:gap_bound}.
	\end{proof}

\end{document}